%% file: main.tex
\documentclass[letterpaper,twocolumn]{article}
\usepackage[margin=.8in]{geometry}
\usepackage{xcolor}
\usepackage[T1]{fontenc}
\usepackage{times}
\usepackage{hyperref}
\usepackage{natbib}
\usepackage{amsmath}
\usepackage{amssymb}
\usepackage{amsthm}
\usepackage{mleftright}
\usepackage{adjustbox}
\usepackage{enumitem}
\usepackage{pifont}
\newcommand{\cmark}{\ding{51}}
\newcommand{\xmark}{\ding{55}}
\usepackage{makecell}
\usepackage{colortbl}
\theoremstyle{definition}
\newtheorem{theorem}{Theorem}
\newtheorem{definition}{Definition}

\newtheorem{corollary}{Corollary}
\usepackage{float}
\usepackage{algorithm}
\usepackage[noend]{algpseudocode}
\usepackage{caption}
\usepackage{authblk}
\usepackage{listings}
\DeclareMathOperator*{\argmax}{argmax}
\DeclareMathOperator*{\argmin}{argmin}
\DeclareMathOperator*{\softmax}{softmax}
\DeclareMathOperator*{\expect}{\mathbb{E}}
\DeclareMathOperator*{\mix}{mix}

\newcommand{\replace}[3]{#1[#2 / #3]}

\title{Simultaneous incremental support adjustment and metagame solving: An equilibrium-finding framework for continuous-action games}

\author[1]{Carlos Martin}
\author[1,2,3,4]{Tuomas Sandholm}
\affil[ ]{\{cgmartin, sandholm\}@cs.cmu.edu}
\affil[1]{Carnegie Mellon University}
\affil[2]{Strategy Robot, Inc.}
\affil[3]{Optimized Markets, Inc.}
\affil[4]{Strategic Machine, Inc.}
\date{}

\begin{document}
\maketitle
\input{abstract}
\input{introduction}
\input{formulation}
\input{prior}
\input{method}
\input{experiments/main}
\input{conclusion}
\input{acknowledgments}
\bibliographystyle{plainnat}
\bibliography{dairefs,references}
\appendix
\input{additional_related_work}
\input{additional_figures}
\input{theoretical/main}

\input{code}
\end{document}

%% file: abstract.tex
\begin{abstract}
We present a framework for computing approximate mixed-strategy Nash equilibria of continuous-action games.
It is a modification of the traditional double oracle algorithm, extended to multiple players and continuous action spaces.
Unlike prior methods, it maintains fixed-cardinality pure strategy sets for each player.
Thus, unlike prior methods, only a constant amount of memory is necessary.
Furthermore, it does not require exact metagame solving on each iteration, which can be computationally expensive for large metagames.
Moreover, it does not require global best-response computation on each iteration, which can be computationally expensive or even intractable for high-dimensional action spaces and general games.
Our method incrementally reduces the exploitability of the strategy profile in the finite metagame, pushing it toward Nash equilibrium.
Simultaneously, it incrementally improves the pure strategies that best respond to this strategy profile in the full game.
We evaluate our method on various continuous-action games, showing that it obtains approximate mixed-strategy Nash equilibria with low exploitability.
\end{abstract}

%% file: introduction.tex
\section{Introduction}

Research on computing \emph{Nash equilibria (NE)} in games has mostly focused on settings with finite, discrete action spaces.
However, many games involving space, money, or time have continuous action spaces. 
These include security games in continuous spaces \citep{Kamra_2017,Kamra_2018,Kamra_2019}, resource allocation games \citep{Ganzfried_2021}, network games \citep{Ghosh_2019}, simulations of military scenarios and wargaming \citep{Marchesi20:Learning}, and video games \citep{Berner19:Dota,Vinyals19:Grandmaster}.
Also, even if the action space is discrete, it may be fine-grained enough to treat as continuous in order to improve the computational efficiency of equilibrium finding~\citep{Borel38:Traite,Chen06:Mathematics,Ganzfried10:Computing}.
The usual approach to computing an equilibrium of a game with continuous action spaces involves discretizing the action space. 
That entails loss in solution quality~\citep{Kroer15:Discretization}. 
Also, it does not scale well; for one, in multidimensional action spaces it entails a combinatorial explosion of discretized points (\emph{i.e.}, exponential in the number of dimensions).
Therefore, other approaches are called for.

We present \emph{Simultaneous Incremental Support Adjustment and Metagame Solving (SISAMS)}, a framework for computing approximate mixed-strategy NE in continuous-action games.
It is a modification of the double oracle algorithm, extended to multiple players and continuous action spaces.
Unlike prior methods, it runs in constant memory, because it maintains fixed-cardinality pure strategy sets for each player.
Furthermore, it does not require exact metagame solving on each iteration, which can be computationally expensive for large metagames.
Moreover, it does not require global best-response computation on each iteration, which can be computationally expensive or even intractable for high-dimensional action spaces and general games.
The method incrementally reduces the exploitability of the strategy profile in the finite metagame, pushing it toward NE.
Simultaneously, it incrementally improves the pure strategies that best respond to this strategy profile in the full game.
We evaluate the method on various continuous-action games, showing that it obtains approximate mixed-strategy NE with low exploitability.

In \S\ref{sec:formulation}, we introduce relevant notation and present a mathematical formulation of the problem we are tackling.
In \S\ref{sec:prior}, we describe prior work.
In \S\ref{sec:method}, we present our method.
In \S\ref{sec:experiments}, we describe the experimental settings we use as benchmarks, and present our experimental results for them.
In \S\ref{sec:conclusion}, we present our conclusions and suggest directions for future research.
In the appendix, we present additional related work, additional figures, theoretical analysis, and code.

%% file: formulation.tex
\section{Problem formulation}
\label{sec:formulation}

We use the following notation.
Between vectors or matrices, \(\odot\) denotes the element-wise or Hadamard product.
If \(n \in \mathbb{N}\), \([n] = \{0, \ldots, n - 1\}\).
If \(\mathcal{X}\) is a set, \(\triangle \mathcal{X}\) is the set of all Borel probability measures on \(\mathcal{X}\).
If \(\mu \in \triangle \mathcal{X}\) and \(f : \mathcal{X} \to \mathcal{Y}\), \(f \# \mu \in \triangle \mathcal{Y}\) is the pushforward of \(\mu\) under \(f\).
If \(x \in \mathcal{X}\), \(\delta(x)\) is the Dirac measure centered at \(x\).
If \(f : \mathcal{X} \to \mathcal{Y}, x \in \mathcal{X}, y \in \mathcal{Y}\), \(f[x/y] : \mathcal{X} \to \mathcal{Y}\) is the function defined by \(f[x/y](x') = y\) if \(x' = x\) and \(f(x')\) otherwise.
If \(\mathcal{A}\) is a family indexed by \(\mathcal{I}\), \(\mathcal{A}_\times = \prod_{i \in I} \mathcal{A}_i\).
If \(\mu_i \in \triangle \mathcal{A}_i\) for \(i \in \mathcal{I}\), \(\bigotimes_{i \in \mathcal{I}} \mu_i \in \triangle \mathcal{A}_\times\) is their product measure.

A strategic-form game is a tuple \((\mathcal{I}, \mathcal{S}, u)\) where \(\mathcal{I}\) is a set of players, \(\mathcal{S}_i\) is a set of strategies for \(i \in \mathcal{I}\), and \(u_i : \mathcal{S}_\times \to \mathbb{R}\) is a utility function for \(i \in \mathcal{I}\).
Here, \(\mathcal{S}_\times = \prod_{i \in \mathcal{I}} \mathcal{S}_i\) is the set of strategy profiles.
Given a strategy profile, a \emph{best response (BR)} for a player is a strategy that maximizes its utility given the other players' strategies.
That is, given \(s \in \mathcal{S}_\times\), a BR for \(i \in \mathcal{I}\) is an element of \(\mathcal{B}_i(s) = \argmax_{r_i \in \mathcal{S}_i} u_i(s[i/r_i])\).
An NE is a strategy profile for which each player's strategy is a BR to the other players' strategies.
That is, it is an \(s \in \mathcal{S}_\times\) such that \(s_i \in \mathcal{B}_i(s)\) for all \(i \in \mathcal{I}\).

Given \(s \in \mathcal{S}_\times\), player \(i\)'s regret is \(R_i(s) = \sup_{r_i \in \mathcal{S}_i} u_i(\replace{s}{i}{r_i}) - u_i(s)\).
Let \(\Psi = \sup_{i \in \mathcal{I}} R_i\).
An \(\varepsilon\)-NE is an \(s \in \mathcal{S}_\times\) such that \(\Psi(s) \leq \varepsilon\).
Define the \emph{exploitability} as \(\Phi = \sum_{i \in \mathcal{I}} R_i\).
Then \(\Psi\) and \(\Phi\) are bounded in terms of each other as follows: \(\Phi / |\mathcal{I}| \leq \Psi \leq \Phi \leq \Psi |\mathcal{I}|\).
Furthermore, they are nonnegative and zero precisely at NE.
Therefore, if NE exist, finding them is equivalent to minimizing \(\Psi\) or \(\Phi\).
Consequently, \(\Psi\) and \(\Phi\) are standard measures of ``closeness'' to NE in the literature \citep{Lanctot17:Unified, Lockhart_2019, Walton_2021, Timbers_2022}.

For any game \((\mathcal{I}, \mathcal{S}, u)\), there exists a mixed-strategy game \((\mathcal{I}, \Sigma, \bar{u})\) where \(\Sigma_i = \triangle \mathcal{S}_i\) and \(\bar{u}_i(\sigma) = \expect_{s \sim \bigotimes_{j \in \mathcal{I}} \sigma_j} u_i(s)\).
That is, each player's strategy is a probability measure over its original strategy set (\emph{i.e.}, a mixed strategy), and its utility is the resulting expected utility in the original game.
A mixed-strategy NE is an NE of the mixed-strategy game.\footnote{
More generally, one can consider settings where randomness is a limited resource (\emph{e.g.}, only a limited number of random bits are available to the agent) or where the agent can only mix between a limited number of pure strategies (\emph{i.e.}, its mixed strategy must be \emph{sparse}).
Alternatively, its mixed strategy may be restricted to some class of \emph{representable} distributions, such as an explicit parametric model or implicit density model like a Generative Adversarial Network \citep{Goodfellow_2014,goodfellow2020generative}.
}

A continuous-action game is a game whose strategy sets are subsets of Euclidean space, \emph{e.g.}, \(\mathcal{S}_i \subseteq \mathbb{R}^d\).
The following theorems apply to such games.
\citet{Nash50:Equilibrium} showed that if each \(\mathcal{S}_i\) is nonempty and finite, a mixed-strategy NE exists.
\citet{Glicksberg52:Further} showed that if each \(\mathcal{S}_i\) is nonempty and compact, and each \(u_i\) is continuous, a mixed-strategy NE exists.
\citet{Glicksberg52:Further,Fan_1952,Debreu_1952} showed that if each \(\mathcal{S}_i\) is nonempty, compact, and convex, and each \(u_i\) is continuous and quasiconcave in \(s_i\), a pure strategy NE exists.
\citet{Dasgupta86:Existence} showed that if each \(\mathcal{S}_i\) is nonempty, compact, convex, and \(u_i\) is upper semicontinuous and graph continuous and quasiconcave in \(s_i\), a pure strategy NE exists.
They also showed that if each \(\mathcal{S}_i\) is nonempty, compact, and convex, and \(u_i\) is bounded and continuous except on a subset (defined by technical conditions) and weakly lower semicontinuous in \(s_i\), and \(\sum_{i \in \mathcal{I}} u_i\) is upper semicontinuous, a mixed-strategy NE exists.
\citet{Rosen_1965} proved the uniqueness of a pure NE for continuous-action games under diagonal strict concavity assumptions.
Most equilibrium-finding algorithms in the literature target discrete-action games, raising the question of how to compute equilibria for continuous-action games.

%% file: prior.tex
\section{Related prior work}
\label{sec:prior}

\paragraph{Double oracle}
\citet{McMahan_2003} introduced \emph{Double Oracle (DO)}, an algorithm for computing NE of two-player zero-sum normal-form games, and proved its convergence.
It maintains finite strategy sets for each player, and iteratively extends them with best responses to an equilibrium of the induced finite metagame.
It can be used to find approximate mixed-strategy NE of games with large action spaces.

\citet{Adam_2021} introduced a DO algorithm for finding an NE in \emph{continuous} two-player zero-sum games with compact strategy sets.
In the authors' words, ``The equilibrium of each finite subgame is found by solving a linear program.
The best responses were computed by selecting the best point of a uniform discretization for the one- dimensional problems and by using a mixed-integer linear programming reformulation for the Colonel Blotto games.''
Their method converges to an NE and is guaranteed to recover an approximate equilibrium in finitely-many steps.

\citet{Kroupa_2021} extended this to \(n\) players under the name of \emph{Multiple Oracle (MO)}.
They prove that their method recovers an approximate equilibrium in finitely many iterations, and converges in the Wasserstein distance to an equilibrium of the original continuous game.
The best response computation is based on global solvers for special classes of utility functions.
According to the authors, ``the choice of best response oracle and the method for solving finite subgames should be fine-tuned for every particular class of game [...] our method converges fast when the dimensions of strategy spaces are small and the generated subgames are not large''.
This algorithm is shown in Algorithm~\ref{alg:double_oracle}.
It starts with nonempty finite subsets of the players' action spaces.
On each iteration, performs the following two steps.
First, it computes an equilibrium of the finite metagame consisting of the restriction of the full game to the finite subsets.
Second, for each player, it computes a best response to this equilibrium \emph{in the full game}, and adds it to the player's finite subset.

\begin{algorithm}
\caption{Double Oracle (DO)}
\label{alg:double_oracle}
\begin{algorithmic}
\State \textbf{input} \((\mathcal{I}, \mathcal{S}, u)\) is a game
\State \textbf{input} \(\mathcal{X}_i\) is a nonempty finite subset of \(\mathcal{S}_i\)
\State \textbf{input} \(T \in \mathbb{N}\) is the number of iterations
\For{\(t = 1\) \textbf{to} \(T\)}
    \State \(\sigma \in \text{MixedStrategyNashEquilibria}((\mathcal{I}, \mathcal{X}, u))\)
    \For{\(i \in \mathcal{I}\)}
        \State \(s_i \in \text{PureStrategyBestResponses}((\mathcal{I}, \mathcal{S}, u), \sigma, i)\)
        \State \(\mathcal{X}_i \gets \mathcal{X}_i \cup \{s_i\}\)
    \EndFor
\EndFor
\State \textbf{output} \(\sigma\)
\end{algorithmic}
\end{algorithm}

\citet{Li_2021} extended DO to \(n\)-player general-sum continuous Bayesian games.
They represent agents as neural networks and optimize them using \emph{natural evolution strategies (NES)} \citep{Wierstra_2008,Wierstra_2014}.
To approximate a pure-strategy equilibrium, they formulate the problem as a bi-level optimization and employ NES to implement both inner-loop best response optimization and outer-loop regret minimization.

\paragraph{PSRO}
\citet{Lanctot17:Unified} introduced \emph{policy-space response oracles (PSRO)}.
It generalizes DO to the case where the metagame's choices are policies rather than actions.
It also generalizes FSP.
Unlike previous work, any meta-solver can be plugged in to compute a new meta-strategy.

\citet{mcaleer2020pipeline} introduced \emph{pipeline PSRO (P2SRO)}, a scalable general method for finding approximate NE in large zero-sum imperfect-information games.
P2SRO parallelizes PSRO with convergence guarantees by maintaining a hierarchical pipeline of reinforcement learning workers, each training against the policies generated by lower levels in the hierarchy.

\citet{McAleer_2021} introduced \emph{extensive-form DO (XDO)}, an extensive-form DO algorithm for two-player zero-sum games that is guaranteed to converge to an approximate NE linearly in the number of infostates.
Unlike PSRO, which mixes best responses at the root of the game, XDO mixes best responses at every infostate.
They also introduced \emph{neural XDO (NXDO)}, where the best response is learned through deep RL, and an approximate equilibrium of the metagame is computed via a deep RL method for finding NE, such as NFSP \citep{Heinrich16:Deep} or DREAM \citep{steinberger2020dream}.
NXDO iteratively adds reinforcement learning policies to a population but solves an extensive-form restricted game, which has been shown to be more efficient than solving a matrix-form restricted game as in PSRO.

\citet{McAleer_2022} introduced \emph{anytime DO (ADO)}, a tabular DO algorithm for \emph{two-player zero-sum} games that is guaranteed to converge to an NE while decreasing exploitability from one iteration to the next.
They also introduced \emph{anytime PSRO (APSRO)}, a version of ADO that calculates best responses via reinforcement learning.

\citet{McAleer_2022b} introduced \emph{self-play PSRO (SP-PSRO)}, which adds an approximately optimal stochastic policy to the population in each iteration.
Instead of adding only deterministic best responses to the opponent's least exploitable population mixture, SP-PSRO also learns an approximately optimal stochastic policy and adds it to the population as well.
This is in contrast to prior methods, which might need to add all deterministic policies before converging.
However, as the authors note in the conclusion, SP-PSRO is a normal-form algorithm in that it mixes at the root of the game tree.
\citet{McAleer_2021} showed that this may take an exponential number of iterations to converge to an approximate NE, and introduced XDO and NXDO to address this problem.
The authors leave combining SP-PSRO with XDO and NXDO to future work.
Like APSRO, SP-PSRO is limited to two-player zero-sum games.

\citet{Muller2020A} extend the theoretical underpinnings of PSRO by considering an alternative solution concept called \(\alpha\)-Rank \citep{omidshafiei2019alpha} rather than NE.
They establish convergence guarantees and identify links between NE and \(\alpha\)-Rank.

\citet{marris2021multi} proposed \emph{Joint Policy-Space Response Oracles (JPSRO)}, an algorithm for training agents in n-player, general-sum extensive form games.
They consider \emph{correlated equilibrium (CE)} \citep{Aumann74:Subjectivity} and \emph{coarse correlated equilibrium (CCE)} as solution concepts, rather than NE.

Due to space constraints, we describe additional related work in the appendix.

%% file: method.tex
\section{Proposed method}
\label{sec:method}

The traditional DO algorithm, extended to multiple players and continuous actions in the way described by \citet{Kroupa_2021}, has some limitations.
First, metagame solving on each iteration can be expensive, especially when the metagame is large.
Second, global best response computation on each iteration can be expensive and even intractable, \emph{e.g.}, for high-dimensional spaces and/or general utility functions.
The key ideas behind our algorithm are to improve an approximate equilibrium \emph{gradually} across iterations, to improve approximate best responses \emph{gradually} across iterations, and to do both of these \emph{simultaneously}.
Our algorithm is shown as Algorithm~\ref{alg:proposed}.

\begin{algorithm}
\caption{Simultaneous Incremental Support Adjustment and Metagame Solving (SISAMS)}
\label{alg:proposed}
\begin{algorithmic}
\State \textbf{input} \((\mathcal{I}, \mathcal{S}, u)\) is a game
\State \textbf{input} \(n_i \in \mathbb{N}\) is the number of meta-actions for \(i\)
\State \textbf{input} \(w_i \in \triangle n_i\) is a mixed meta-strategy for \(i\)
\State \textbf{input} \(\mathbf{x}_i : n_i \to \mathcal{S}_i\) is an action mapping for \(i\)
\State \textbf{input} \(\alpha, \beta : \mathbb{N} \to \mathbb{R}_{\geq 0}\) are stepsize schedules
\State \textbf{input} \(T \in \mathbb{N}\) is the number of iterations
\State \(u^\mathbf{x}_i : n_\times \to \mathbb{R}\) is the metagame utility function for \(i\)
\For{\(t = 1\) \textbf{to} \(T\)}
    \State \(\dot{w} = -\nabla_w \text{Exploitability}((\mathcal{I}, n, u^\mathbf{x}), w)\)
    \For{\(i \in \mathcal{I}\)}
        \State \(\dot{\mathbf{x}}_i = \nabla_{\mathbf{x}_i} \mix_{j \in n_i} u^\mathbf{x}_i(w[i / \delta(j)])\)
    \EndFor
    \State \(w \gets w + \dot{w} \alpha_t\)
    \State \(\mathbf{x} \gets \mathbf{x} + \dot{\mathbf{x}} \beta_t\)
\EndFor
\State \textbf{output} \(\{\mathbf{x}_i \# w_i\}_{i \in \mathcal{I}}\)
\end{algorithmic}
\end{algorithm}

Unlike the standard DO algorithm, it maintains fixed-cardinality sets of pure strategies\footnote{In this context, we use ``pure strategy'' and ``action'' interchangeably.}, or \emph{supports} for each player.
It adjusts the weight of each pure strategy to decrease the exploitability of the metagame.
(Though exploitability involves a maximum and is thus not everywhere differentiable, it \emph{is} continuous and almost-everywhere differentiable, so we use its subgradient.)
Simultaneously, it adjusts the pure strategies that are best responses to the current approximate equilibrium to increase the corresponding player's utility against that equilibrium.
Unlike standard DO, it does not require exact metagame solving on each iteration, which can be computationally expensive for large metagames.
Furthermore, unlike standard DO, it does not require global best response computation on each iteration, which is computationally expensive or even intractable for high-dimensional action spaces and/or general utility functions.
Standard DO algorithms maintain a set of strategies that is expanded on each iteration with approximate best responses to the meta-strategies of the other players.
These strategies are \emph{static}, that is, they do not change after they are added.
In contrast, our algorithm \emph{dynamically} improves the elements of a best-response ensemble during training.
Thus the ensemble does not need to be grown with each iteration, but improves autonomously over time.

Instead of improving \emph{only} the \emph{best} response on each iteration, we can incentivize the other pure strategies to improve as well.
To do this, we introduce the \emph{rank-based mixing operatator}, or \(\mix\).
It is defined as \(\mix_{i \in \mathcal{I}} x_i = \frac{1}{|\mathcal{I}|} \sum_{i \in \mathcal{I}} r_i x_i\) where \(r_i \in \{1, \ldots, |\mathcal{I}|\}\) is the ordinal rank of element \(i\).
This gives all elements some weight, but more weight to those that perform better.
Thus it gives \emph{all} elements some chance to improve, while simultaneously incentivizing specialization, that is, preventing uniformity.
The intuition behind this is that an individual pure strategy does not need to perform well against all counter-strategies.
Rather, \emph{the ensemble as a whole} does.
Because what matters is the ensemble of pure strategies, an individual pure strategy does not need to perform well against all counter-strategies, but only a subset thereof.
Since the best-performing pure strategies will be chosen automatically as responses, intuitively, a ``good'' ensemble will have pure strategies that are specialized.
The weight of each element depends only on its rank, so the weights are invariant under monotone transformations of the input values.
Thus \(\mix\) is shift-equivariant: \(\mix(\mathbf{x} + \alpha \mathbf{1}) = \mix(\mathbf{x}) + \alpha\).
It is also scale-equivariant: \(\mix(\alpha \mathbf{x}) = \alpha \mix(\mathbf{x})\) for \(\alpha \geq 0\).

When the pure strategy set sizes of the players are all 1, the metagame is trivial, and our algorithm becomes equivalent to simultaneous gradient dynamics, which is a local, differential version of best-response dynamics.
This can be seen from the fact that, in that case, the algorithm reduces to performing \(\mathbf{x} \gets \mathbf{x} + \dot{\mathbf{x}} \beta_t\) where \(\dot{\mathbf{x}}_i = \nabla_{\mathbf{x}_i} u_i^\mathbf{x}(\bullet) = \nabla_{\mathbf{x}_i} u_i(\mathbf{x})\) for each player \(i \in \mathcal{I}\).

Our algorithm uses gradients of the utility functions. 
Some continuous-action games have discontinuous utility functions.
To handle those cases, we replace the corresponding gradient with a \emph{pseudogradient}, which is the gradient of a smoothened version of the function (\emph{e.g.}, the function convolved with a narrow Gaussian).
An unbiased estimator for this pseudogradient is obtained by evaluating the function at randomly-sampled perturbed points and using their values to approximate directional derivatives along those directions \citep{Duchi_2015, Nesterov_2017, Shamir_2017, Salimans_2017, Berahas_2022, metz2021gradients}.
For example, \(\nabla_\mathbf{x} \expect_{\mathbf{z} \sim \mathcal{N}} f(\mathbf{x} + \sigma \mathbf{z}) = \expect_{\mathbf{z} \sim \mathcal{N}} \frac{1}{\sigma} f(\mathbf{x} + \sigma \mathbf{z}) \mathbf{z}\) where \(\mathcal{N}\) is the standard multivariate normal distribution with same dimension as \(\mathbf{x}\).
Table \ref{tab:comparison} summarizes some of the main differences between our method and prior ones from the literature.

\newcommand{\T}{\cellcolor{green!25}\cmark}
\newcommand{\F}{\cellcolor{red!25}\xmark}
\begin{table*}
    \centering
    \begin{tabular}{lcccccccc}
        \hline
         & \(\alpha\)PSRO & JPSRO & NXDO & APSRO & SP-PSRO & MO & RFP & \textbf{SISAMS} \\
        \hline
        \(n\) players & \T & \T & \F & \F & \F & \T & \T & \T \\
        No BR oracle & \F & \T & \T & \T & \T & \F & \F & \T \\
        No equilibrium oracle & \F & \F & \T & \T & \T & \F & \T & \T \\
        NE solution concept & \F & \F & \T & \T & \T & \T & \T & \T \\
        Constant memory & \F & \F & \F & \F & \F & \F & \F & \T \\
        \hline
    \end{tabular}
    \caption{Qualitative comparison of our method to prior methods from the literature.}
    \label{tab:comparison}
\end{table*}

%% file: experiments/main.tex
\section{Experiments}
\label{sec:experiments}

To facilitate gradient-based optimization, in games that have bounded action spaces, we use \emph{reparameterization}.
That is, we parameterize the actual bounded actions in terms of an unbounded parameter space using smooth (\(C^\infty\)) ``squeezing'' functions.
Examples include
\(\tanh : \mathbb{R} \to (-1, 1)\),
\(\operatorname{sigmoid} : \mathbb{R} \to (0, 1)\),
\(\operatorname{softplus} : \mathbb{R} \to (0, \infty)\),
\((x \mapsto x^2) : \mathbb{R} \to [0, \infty)\),
and \(\operatorname{softmax} : \mathbb{R}^d \to \triangle [d]\).
Reparameterization is also helpful for pseudogradients, since otherwise perturbations can land outside of the valid domain of the utility function.

Some of the games we test on are based on games with discontinuous utility functions.
Although we can tackle these using pseudogradients, as described in Section \ref{sec:method}, this is an implementation detail and not the focus of this paper.
Therefore, we instead modify the discontinuous utility functions to smoothen them.
For example, we replace ``hard'' argmaxes with softmaxes.
The softmax function is defined as
\(\operatorname{softmax}(\mathbf{x})_i = \frac{\exp \mathbf{x}_i}{\sum_{j \in [n]} \exp \mathbf{x}_j}\).
It satisfies the property that
\(\operatorname{softmax}(\beta \mathbf{x}) \to \operatorname{onehot}(\argmax(\mathbf{x}))\) as \(\beta \to \infty\).
Here, \(\beta \in \mathbb{R}\) can be interpreted as an ``inverse temperature'' or ``sharpness'' parameter.
There is also a probabilistic interpretation:
\(\operatorname{softmax}(\mathbf{x}) = \operatorname{E}_{\mathbf{z} \sim \text{Gumbel}} \operatorname{onehot}(\argmax(\mathbf{x} + \mathbf{z}))\).
Thus using softmax is equivalent to using argmax but perturbing the inputs with Gumbel noise scaled by \(\beta^{-1}\), followed by an expectation.

\input{experiments/interval}

\input{experiments/circle}

\input{experiments/glicksberg_gross}

\input{experiments/blotto}

\input{experiments/security}

\input{experiments/auction}

\input{experiments/chopstick}

\paragraph{Experimental results}
In our experiments, we use \(\beta = 20\), \(10^5\) iterations per epoch, and \(16\) trials per experiment.
Figures~\ref{fig:exploitabilities_1} and \ref{fig:exploitabilities_2} show exploitabilities for the various games we test on.
In these plots, solid lines show the mean across trials, and bands show its standard error.
In each legend, ``lr'' means the learning rate and ``size'' means the size of each player's support.
On the X axis, a ``batch'' is a single iteration of the algorithm.
Each experiment ran on one NVIDIA A100 SXM4 40GB GPU on a computer cluster.
The appendix contains additional figures illustrating equilibria for these games and what the strategies learned by our algorithm look like.

The exploitability plots show that our algorithm converges to low exploitability when the sizes of the players' supports are sufficiently large, with larger supports allowing for a better approximation to an NE for the underlying game.
As expected, a support size of 1, \emph{i.e.}, having a single pure strategy in each player's strategy set, causes a failure to converge to low exploitability.
Thus having multiple pure strategies to mix over is crucial for performance.

\begin{figure*}
    \centering
    \includegraphics[width=.5\linewidth]{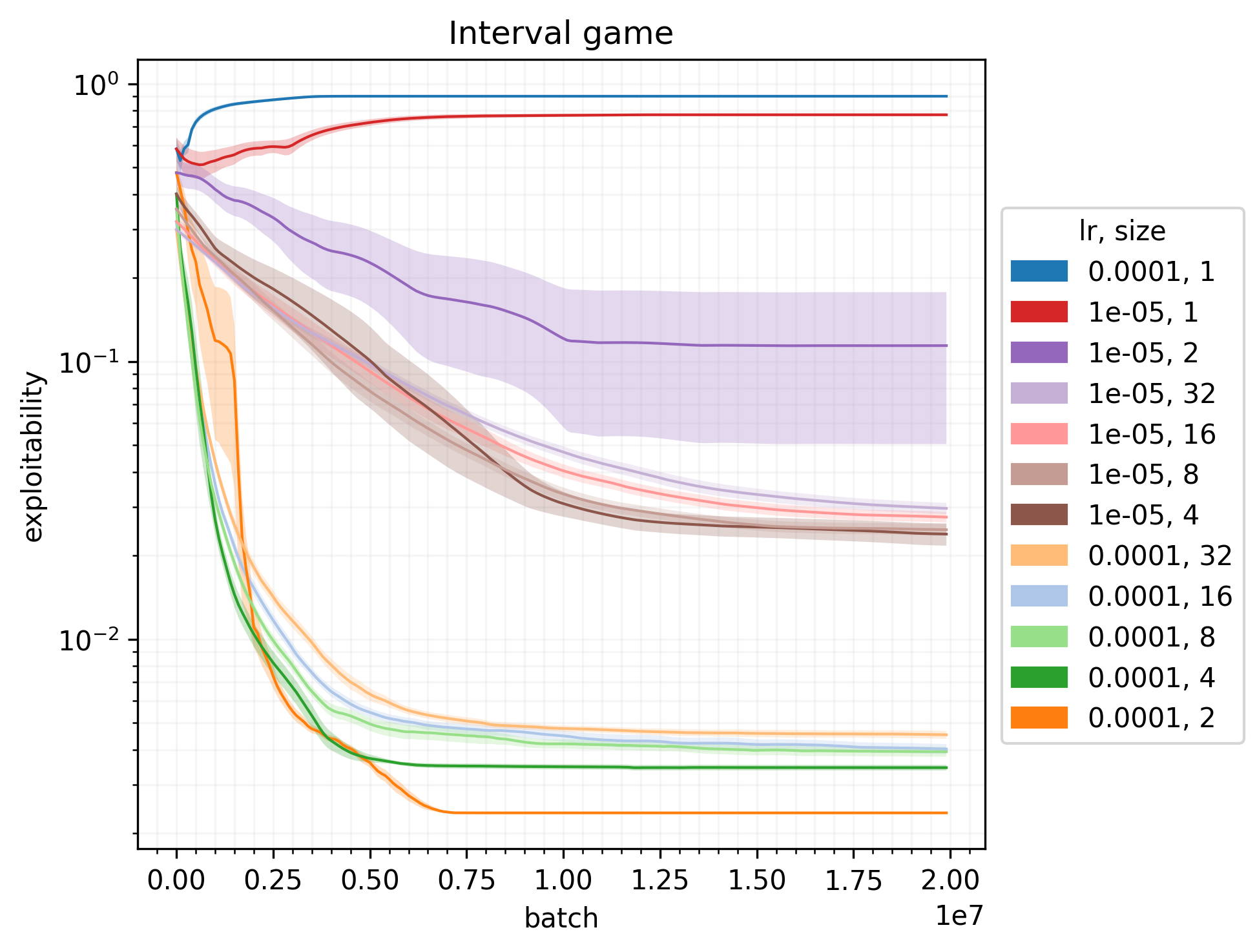}%
    \includegraphics[width=.5\linewidth]{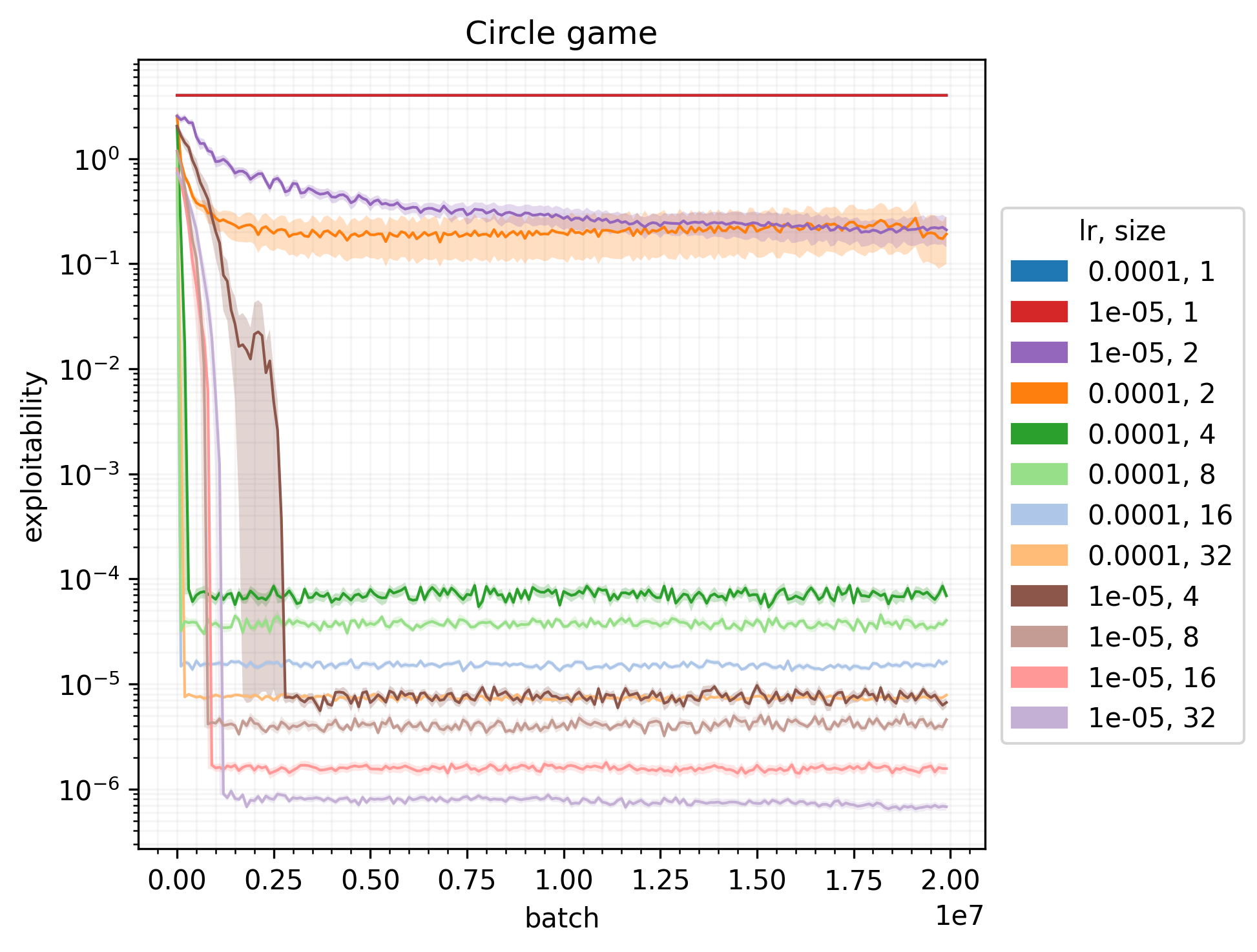}
    \includegraphics[width=.5\linewidth]{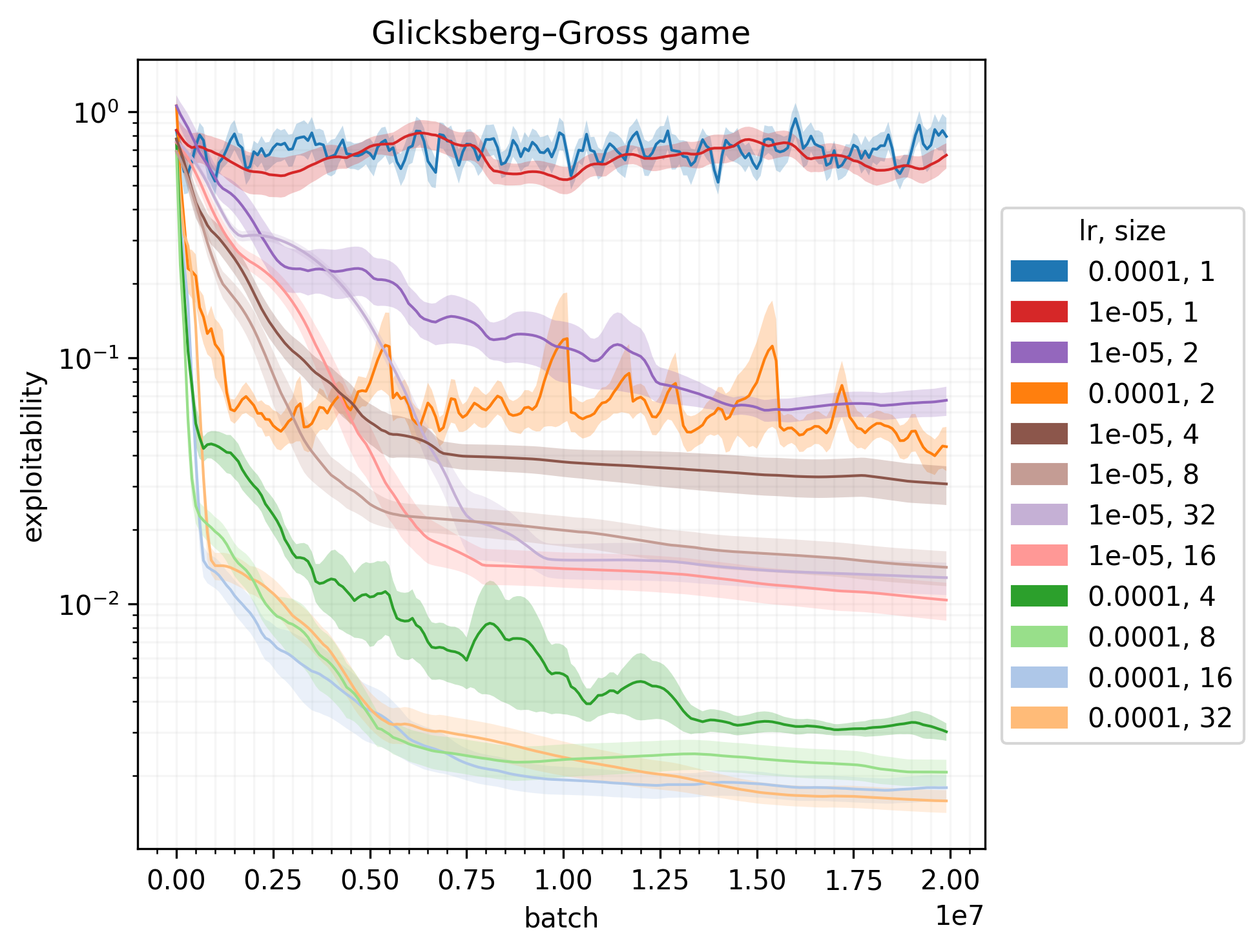}%
    \includegraphics[width=.5\linewidth]{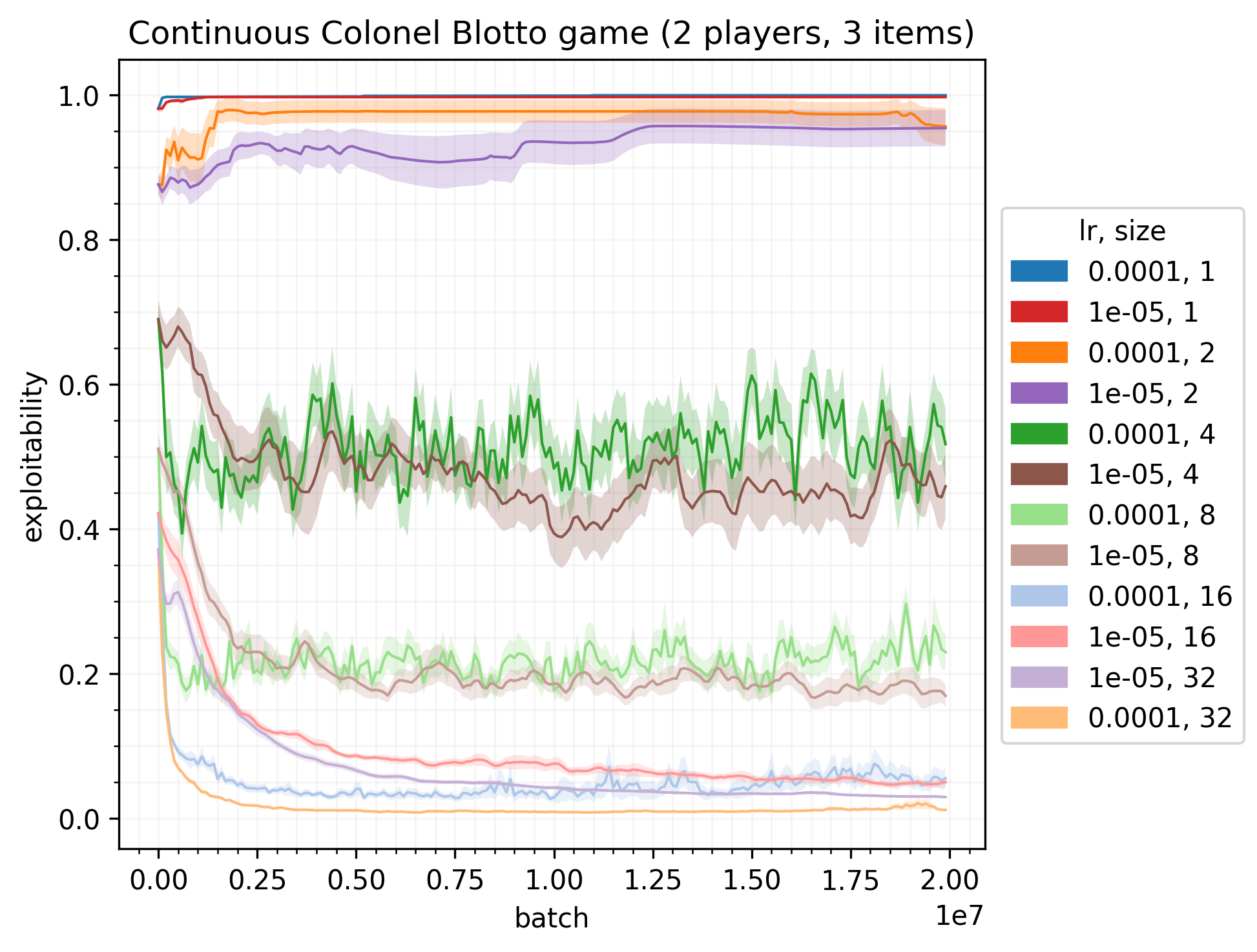}
    \includegraphics[width=.5\linewidth]{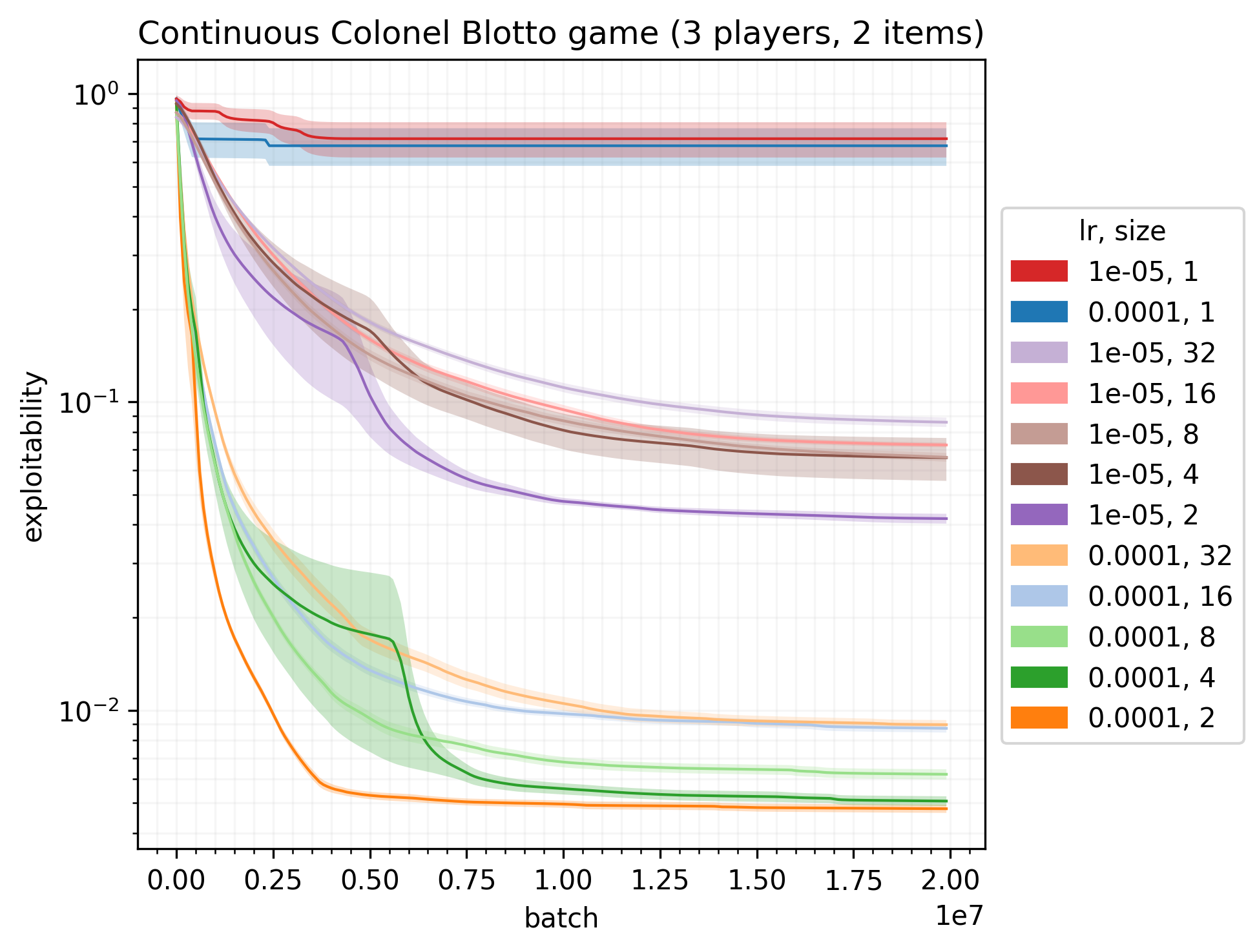}%
    \includegraphics[width=.5\linewidth]{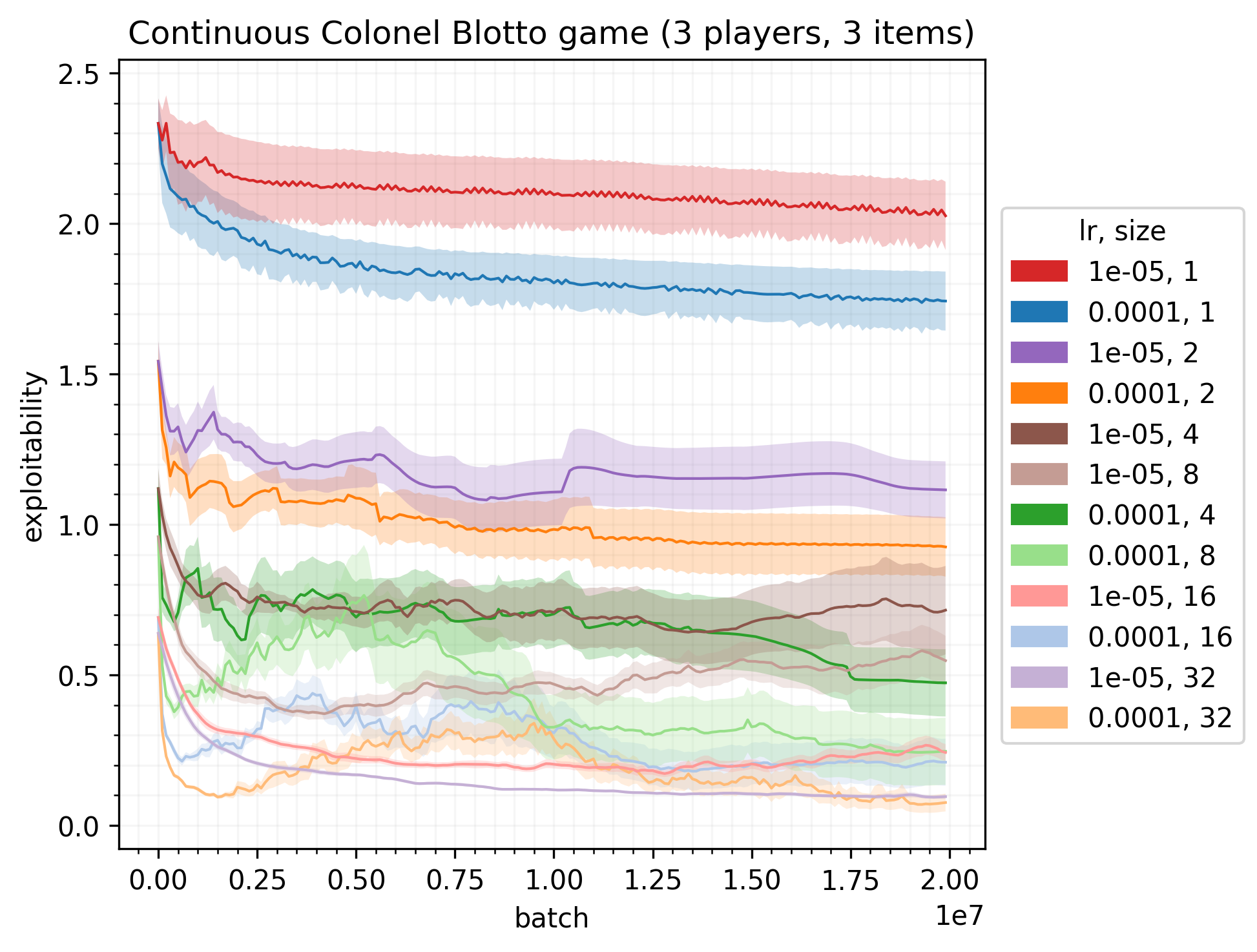}
    \caption{Exploitabilities of learned mixed strategies (part 1).}
    \label{fig:exploitabilities_1}
\end{figure*}

\begin{figure*}
    \centering
    \includegraphics[width=.5\linewidth]{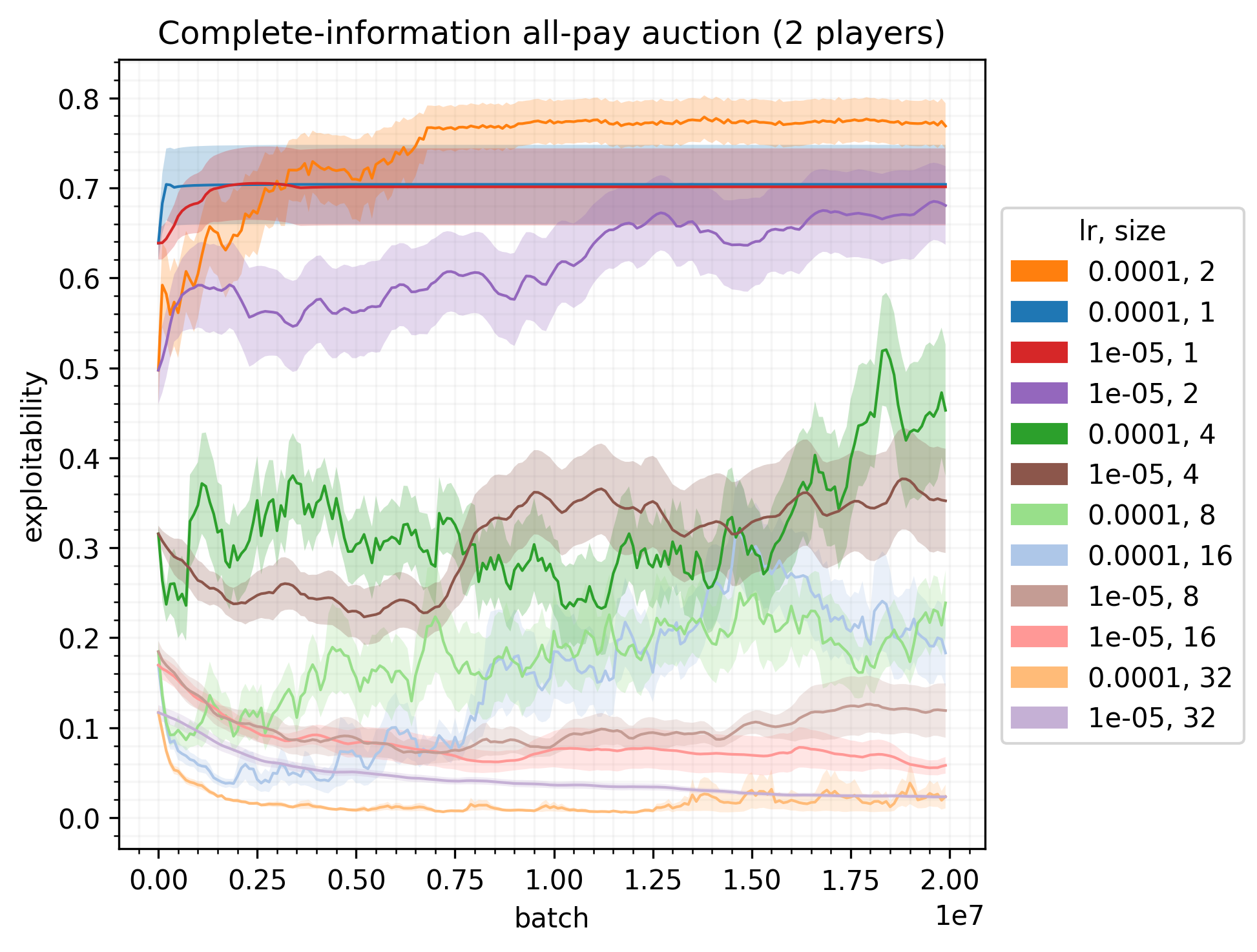}%
    \includegraphics[width=.5\linewidth]{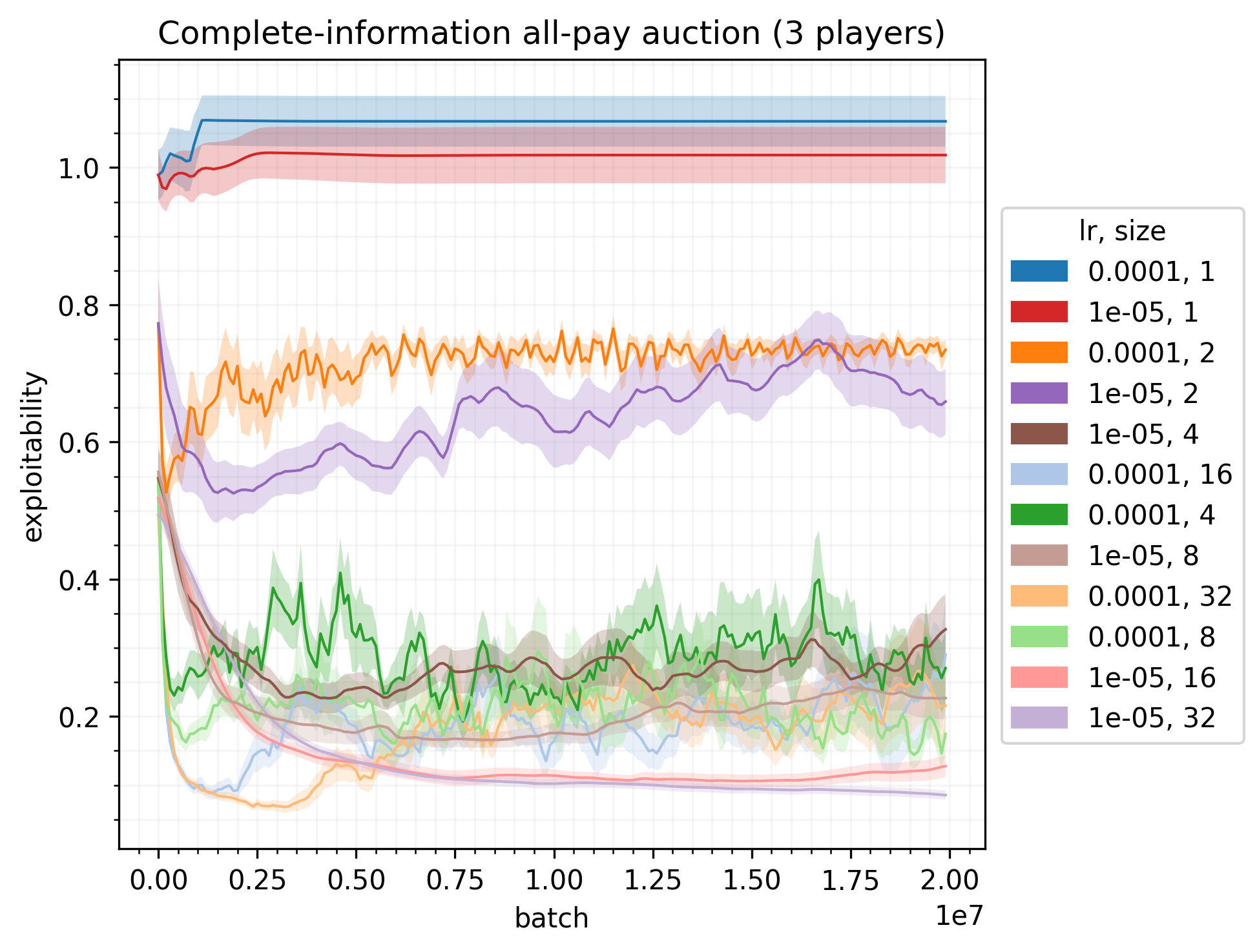}
    \includegraphics[width=.5\linewidth]{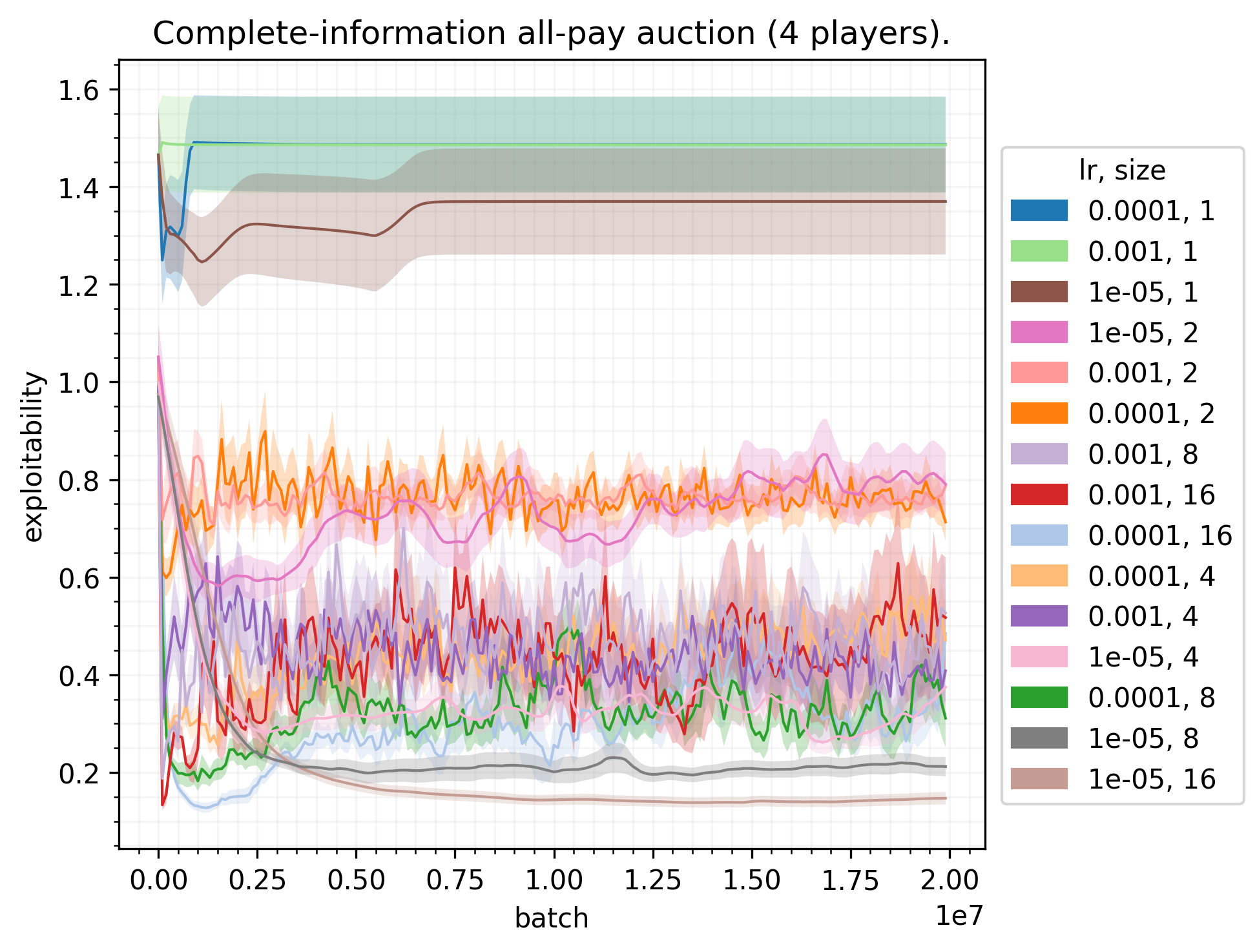}%
    \includegraphics[width=.5\linewidth]{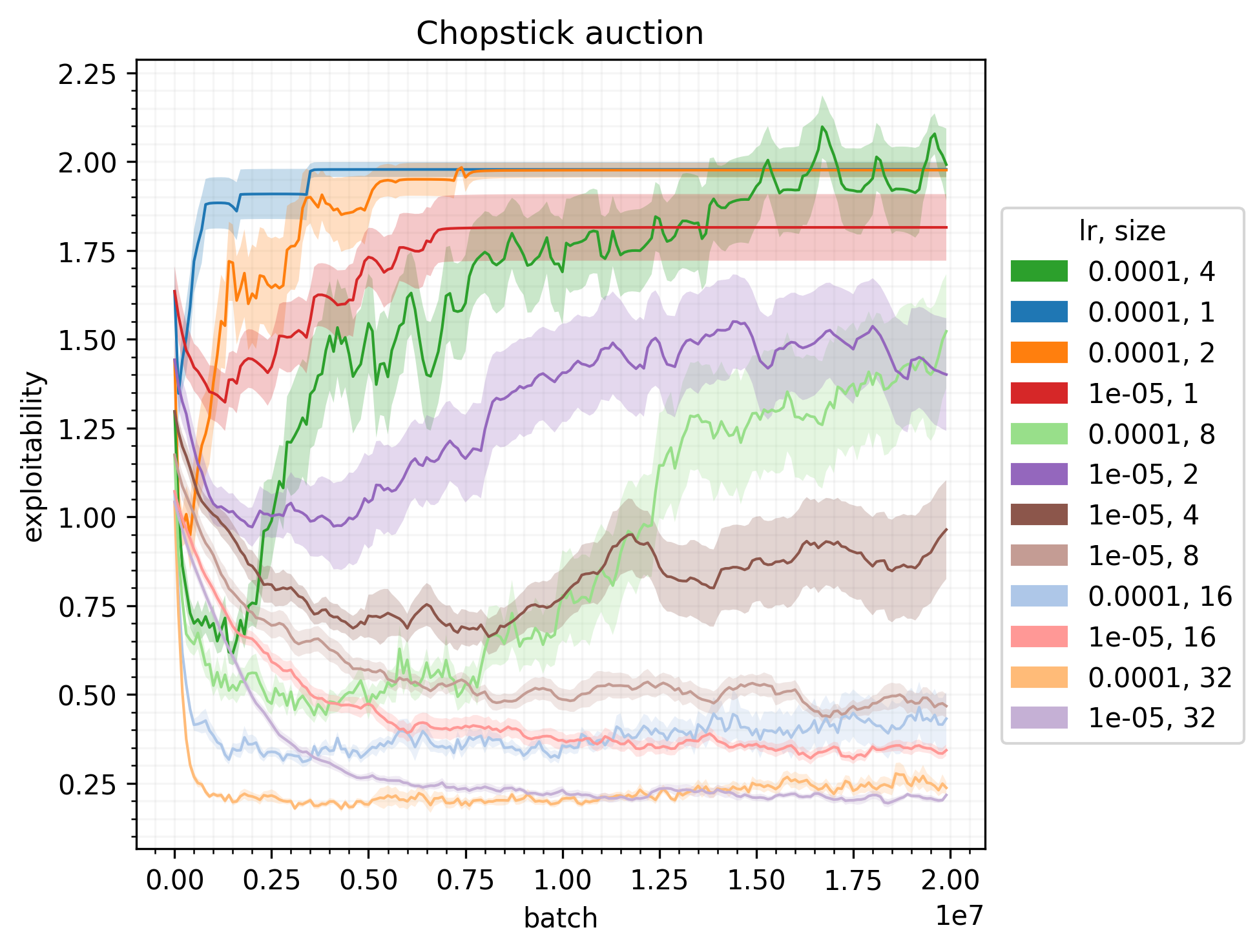}
    \includegraphics[width=.5\linewidth]{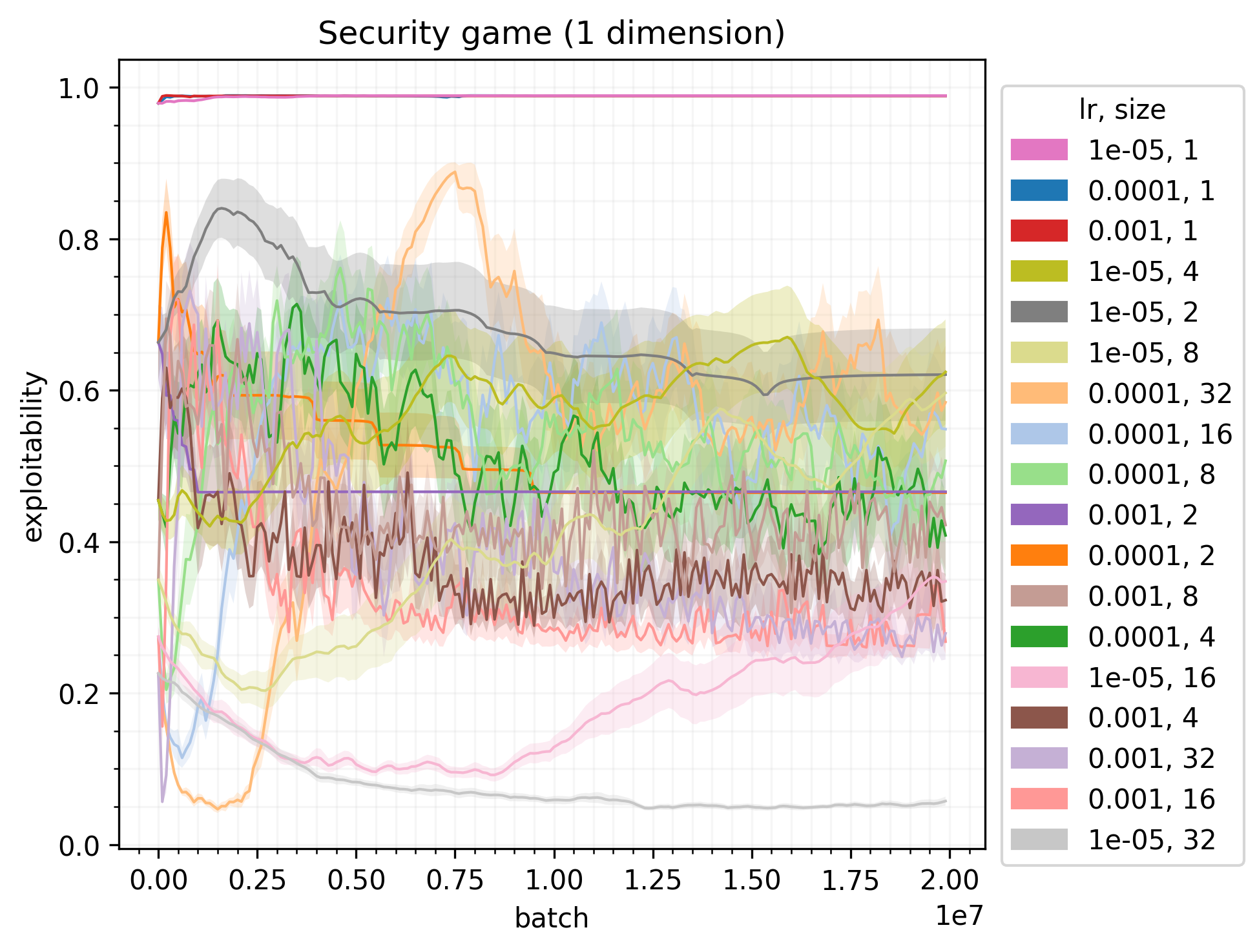}%
    \includegraphics[width=.5\linewidth]{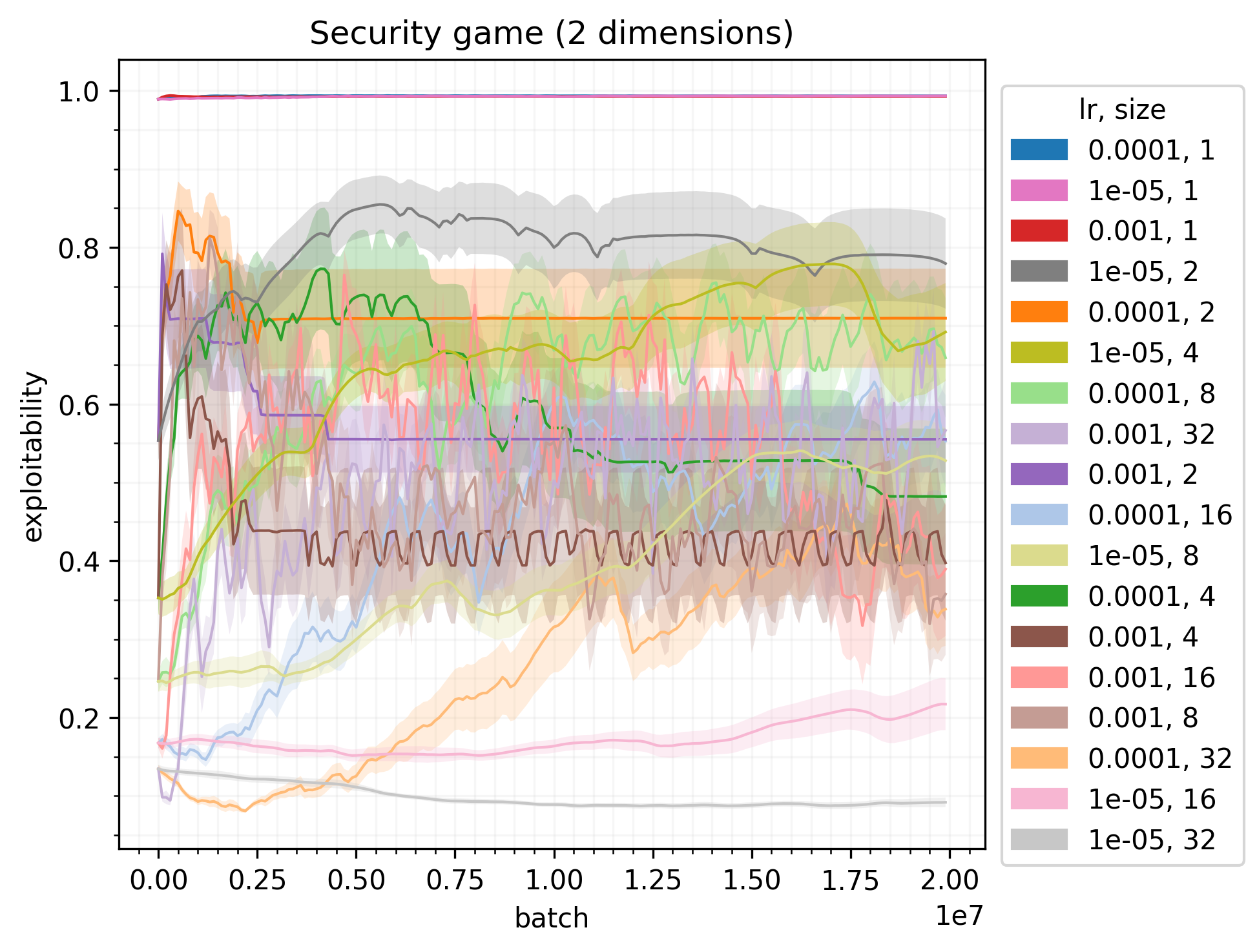}
    \caption{Exploitabilities of learned mixed strategies (part 2).}
    \label{fig:exploitabilities_2}
\end{figure*}

%% file: experiments/interval.tex
\paragraph{Interval game}
Among the two-player games that have infinite sets of pure strategies available to the contestants, ``games on the square'' have received much attention \citep{Kuhn_Tucker_1953, gross1957rational, Parthasarathy_1970, parthasarathy1975equilibria, Chin_1976, szep1985games, stein2008separable}.
These games generalize matrix games by replacing the elements \(M_{ij}\) of a utility matrix by a real-valued function \(M(x, y)\) defined on \(x, y \in [a, b]\) for some \(a, b \in \mathbb{R}, a < b\).
The pure strategies available to players 1 and 2 are indexed by the real numbers \(x\) and \(y\), while their mixed strategies become probability distributions over the interval.

An example of such a game is the following.
Let \(\mathcal{I} = \{1, 2\}\), \(\mathcal{S}_i = [-1, 1]\), and \(u_1(x, y) = -u_2(x, y) = (x - y)^2\).
This game was studied by \citet[p. 2856]{parrilo2006polynomial}.
Since Player 2 wants to minimize the squared distance, it should try to ``guess'' the number chosen by Player 1.
Conversely, the first player should try to make its number as difficult to guess as possible.
It has an NE where Player 1 picks \(\pm 1\) with equal probability and Player 2 picks \(0\).
Player 1's expected utility is \(1\).
Our algorithm converges when Player 1 has at least two pure strategies to randomize over, which is expected since that is the number of pure strategies required by their equilibrium strategy.

%% file: experiments/circle.tex
\paragraph{Circle game}
This is the same as the interval game but on the unit circle, which is both compact and homogeneous.
More precisely, let \(\mathcal{I} = \{1, 2\}\), \(\mathcal{S}_i = \mathbb{S} \subset \mathbb{R}^2\) be the unit circle, and \(u_1(x, y) = -u_2(x, y) = \|x - y\|_2^2\).
In equilibrium, each player picks uniformly on the unit circle.

%% file: experiments/glicksberg_gross.tex
\paragraph{Glicksberg--Gross game}
This is a game on the square that was introduced and analyzed by \citet{Glicksberg_1953}.
It has utility function \(u : [0, 1]^2 \to \mathbb{R}^2\) where \(u(x, y)_1 = -u(x, y)_2 = \frac{(1 + x) (1 + y) (1 - x y)}{(1 + x y)^2}\).
It has a unique mixed-strategy NE where the strategy of each player has cumulative distribution function \(F(t) = \frac{4}{\pi} \arctan \sqrt{t}\) for \(t \in [0, 1]\), and Player 1's utility is \(\frac{4}{\pi}\).
The authors use this example to illustrate that, unlike games with polynomial utility functions, games with utility functions that are quotients of polynomials do not always have equilibria that are concentrated on a discrete set of points.

%% file: experiments/blotto.tex
\paragraph{Continuous Colonel Blotto game}
A Colonel Blotto game is a game in which players distribute limited resources over multiple battlefields (\emph{i.e.}, items).
A battlefield is won by whoever devotes the most resources to it.
A player's utility is the number of battlefields they win.
It models real-world situations of conflict or competition that involve \emph{resource allocation}, such as political campaigns, research and development, national security, and systems defense.
The original Colonel Blotto game was introduced and studied by \citet{Borel_1953}.
Subsequently, many variants have been introduced and studied in the literature.

\citet{Gross_1950} analyzed a continuous variant in which both players have continuous, possibly unequal budgets.
They obtained exact solutions for various special cases, including all 2-battlefield cases and all 3-battlefield cases with equal budgets.
\citet{Washburn_2013} generalized to the case where battlefield values are unequal across battlefields.
\citet{Kovenock_2021} generalized to the case where battlefield values are also unequal across players.
\citet{Adamo_2009} studied a variant in which players have incomplete information about the other player's resource budgets.
\citet{Kovenock_2011} studied a model where the players are subject to incomplete information about the battlefield valuations.
\citet{Adsera_2021} analyzed the natural multiplayer generalization of the continuous Colonel Blotto game.

We consider a continuous Colonel Blotto game with fixed homogeneous budgets and valuations.
\citet{Gross_1950} analyzed two-player zero-sum case with three battlefields, and derived an analytical equilibrium strategy, which they describe geometrically as follows:
``[The player] inscribes a circle within [the triangle] and erects a hemisphere upon this circle.
He next chooses a point from a density uniformly distributed over the surface of the hemisphere and projects this point straight down into the plane of the triangle...
He then divides his forces in respective proportion to the triangular areas subtended by [this point] and the sides.''
Formally, our game's utility function is
\(u : (\triangle [k])^n \to \mathbb{R}^n\)
where
\(n\) is the number of players,
\(k\) is the number of battlefields,
and \(u(\mathbf{A}) = \softmax(\beta \mathbf{A}) \cdot \mathbf{1}_k\).

%% file: experiments/security.tex
\paragraph{Security game}
In game theory, a coordination game is one that rewards players choosing the same action.
In contrast, an \emph{anti}-coordination game \emph{penalizes} choosing the same action.
A \emph{dis}coordination game is a hybrid of the two.
Specifically, one player is incentivized to choose the same action as the other player, while the other player is incentivized to choose a different action.
Such games, in general, lack pure-strategy NE.
The canonical example of such a game is the \(2 \times 2\) \emph{matching pennies} game, whose utility matrix is (a positive affine transformation of) the identity matrix.
To extend this game to a continuous action space, we replace the notion of choosing the \emph{same} action with that of choosing a \emph{nearby} action, under some metric and threshold distance.
More precisely, let \(\mathcal{I} = \{1, 2\}\), \(\mathcal{S}_i = [0, 1]\), and \(u_1(x, y) = -u_2(x, y) = (1 + (x - y)^2 / d^2)^{-1}\).
Here, \(d\) is a distance scale (we use \(0.1\)).
This game has the form of a security game on a continuous space.
Such games have been studied by \citet{Kamra_2017,Kamra_2018,Kamra_2019}.

%% file: experiments/auction.tex
\paragraph{All-pay auction}
An auction is a mechanism by which one or more items are sold to one or more bidders.
Auctions play a central role in the study of markets and are used in a wide range of contexts.
In a single-item sealed bid auction, bidders simultaneously submit bids and the highest bidder wins the item.
In an \emph{all-pay} auction, each player always pays its bid.
Such auctions are widely used to model lobbying for rents in regulated and trade protected industries, technological competition and R\&D races, political campaigns, job promotions, and other contests \citep{Baye_1996}.
In an \emph{complete-information} all-pay auction, the valuations of the players for all items are common knowledge.
Such an auction lacks pure-strategy equilibria \citep{Baye_1996}.
In the two-player case, in equilibrium, the players bid uniformly on \([0, 1]\).
Formally, our game's utility function is
\(u : [0, 1]^n \to \mathbb{R}^n\)
where
\(n\) is the number of players
and \(u(\mathbf{a}) = \softmax(\beta \mathbf{a}) - \mathbf{a}\).

%% file: experiments/chopstick.tex
\paragraph{Chopstick auction}
Multi-item auctions are of great importance in practice, for example in strategic sourcing \citep{Sandholm13:Very} and radio spectrum allocation \citep{Milgrom14:Deferred,Milgrom_2020}.
However, deriving equilibrium bidding strategies for multi-item auctions is notoriously elusive.
A rare notable instance where equilibrium strategies have been derived is the \emph{chopstick auction} \citep{Szentes_2003,Szentes03:Chopsticks}.
In this auction, 3 chopsticks are sold simultaneously in separate first-price sealed-bid auctions.
There are 2 bidders, and it is common knowledge that a pair of chopsticks is worth \$1, a single chopstick is worth nothing by itself, and 3 chopsticks are worth the same as 2.
Here, pure strategies are triples of non-negative real numbers (bids).

This game was analyzed by \citet{Szentes_2003, Szentes03:Chopsticks} and has an interesting equilibrium.
Specifically, let \(T\) be the regular tetrahedron that is the convex hull of the four points \((\tfrac{1}{2}, \tfrac{1}{2}, 0)\), \((\tfrac{1}{2}, 0, \tfrac{1}{2})\), \((0, \tfrac{1}{2}, \tfrac{1}{2})\), and \((0, 0, 0)\).
Then the uniform distribution on the 2-dimensional surface of \(T\) generates a symmetric equilibrium.
Furthermore, all points inside the tetrahedron are pure best responses to this equilibrium mixture.
We benchmark on the chopstick auction since it is a rare case of a multi-item auction with a known analytic equilibrium, so we can compare our output to an exact equilibrium.
It is also a canonical case of simultaneous separate auctions under combinatorial preferences.
Formally, our game's utility function is
\(u : [0, 1]^2 \to \mathbb{R}^2\)
where
\(u(\mathbf{A}) = \operatorname{sigmoid}(\beta (\mathbf{W} \cdot \mathbf{1}_3 - 1.5)) - (\mathbf{A} \odot \mathbf{W}) \cdot \mathbf{1}_3\)
and
\(\mathbf{W} = \softmax(\beta \mathbf{A})\).

%% file: conclusion.tex
\section{Conclusion}
\label{sec:conclusion}

We introduced an algorithm for computing approximate mixed-strategy NE of continuous-action games.
It is a modification of the DO algorithm, extended to continuous action spaces and multiple agents.
It gradually improves a metagame equilibrium while simultaneously and gradually improving pure strategies against it.
It mitigates the need for exact equilibrium and/or best-response computation on each iteration.
Empirically, when tested on various continuous-action games, our algorithm obtains approximate mixed-strategy NE with low exploitability.
In some games, the quality of the approximation depends, as expected, on the number of pure strategies available to each player.

For future work, we would like to prove the convergence of this algorithm to (approximate) mixed-strategy NE under various conditions, for various classes of games.
To that end, we could analyze the dynamics of the algorithm in cases where the weights are changing and supports are fixed (or slowly-moving), or vice versa.
An analysis of the former case can be found in the appendix.
We also wish to analyze the convergence of this algorithm on discrete-action games with strategies parameterized with continuous variables.

We would also like to apply our algorithm to multi-step \emph{partially-observable stochastic games (POSGs)} with continuous action spaces, in which information sets or observation histories are not explicitly given and agents must instead learn what to remember from observations at each timestep.  
For this, agents can be given memory, \emph{e.g}., using recurrent neural networks~\citep{Rumelhart_1986, Werbos_1988, Hochreiter_1997, Cho_2014}.

Another extension we intend to explore in future work is ``particle resetting'', in which an element of a support (or ``particle'') can be ``reset'', \textit{e.g.}, if its probability falls below some threshold.
The idea behind this approach is to bring back into play particles that appear to be useless or of low value.
A particle may be reset in a way that is influenced by the locations of the other particles.
For example, it can be brought closer to better-performing ones, while at the same time maintaining some distance from them (and potentially yet different particles) to encourage diversity and exploration.

%% file: acknowledgments.tex
\section{Acknowledgments}

This material is based on work supported by the Vannevar Bush Faculty Fellowship ONR N00014-23-1-2876, National Science Foundation grants RI-2312342 and RI-1901403, ARO award W911NF2210266, and NIH award A240108S001.

%% file: additional_related_work.tex
\section{Additional related work}

\citet{perkins2015mixed} proposed an actor-critic reinforcement learning algorithm that adapts mixed strategies over continuous action spaces, and prove that the continuous dynamics of the process converge to equilibrium in the case of potential games.

\citet{Fichtl_2022} presented an approach that computes distributional strategies \citep{Milgrom85:Distributional} (a form of mixed strategies for a Bayesian game) on a discretized version of the game via online convex optimization, specifically \emph{simultaneous online dual averaging (SODA)}.
They prove that the equilibrium of the discretized game approximates an equilibrium in the continuous game.

\citet{Bichler_2021} presented a learning method that represents strategies as neural networks and applies simultaneous gradient dynamics to provably learn local equilibria.
They focused on symmetric auction models, assuming symmetric prior distributions and symmetric equilibrium bidding strategies.
\citet{Bichler_2022} extended this to asymmetric auctions, where one needs to train multiple neural networks.
The previous two papers restrict their attention to pure strategies.

\citet{ijcai2023p317} tackled continuous-action Bayesian games in generality by modeling mixed strategies using randomized policy networks and applying zeroth-order optimization techniques that combine smoothed gradient estimators with equilibrium-finding dynamics.

\subsection{Abstraction}

\citet{Brown14:Regret} studied optimizing a parameter vector for a player in a two-player zero-sum game (\emph{e.g.}, optimizing bet sizes to use in poker).
They propose a custom gradient descent algorithm that provably finds a locally optimal parameter vector.
It optimizes the parameter vector while simultaneously finding an equilibrium.
This amounts to the first action abstraction algorithm (algorithm for selecting a small number of discrete actions to use from a \emph{continuum} of actions--a key preprocessing step for solving large games using current equilibrium-finding algorithms) with convergence guarantees for extensive-form games.

Abstraction is a key component in solving extensive-form games of incomplete information.
\citet{Kroer14:Extensive} introduced a theoretical framework that can be used to give bounds on solution quality for any perfect-recall extensive-form game.
The framework uses a new notion for mapping abstract strategies to the original game, and it leverages a new equilibrium refinement for analysis.
Using this framework, they develop the first general lossy extensive-form game abstraction method with bounds.

Many real-world domains require modeling with continuous action spaces.
This is usually handled by heuristically discretizing the continuous action space without solution quality bounds.
Leveraging recent results on abstraction solution quality, \citet{Kroer15:Discretization} developed the first framework for providing bounds on solution quality for discretization of continuous action spaces in extensive-form games.

\citet{Brown15:Simultaneous} introduced a method that combines abstraction with equilibrium finding by enabling actions to be added to the abstraction at run time.
This allows an agent to begin learning with a coarse abstraction, and then to strategically insert actions at points that the strategy computed in the current abstraction deems important.
The algorithm can quickly add actions to the abstraction while provably not having to restart the equilibrium finding.
It enables anytime convergence to an NE of the full game even in infinite games.

\citet{Kroer16:Imperfect} presented the first general, algorithm-agnostic, solution quality guarantees for NE and approximate self-trembling equilibria computed in imperfect-recall abstractions, when implemented in the original (perfect-recall) game.

\subsection{Fictitious play}

\citet{Brown51:Iterative} introduced \emph{fictitious play (FP)}, a popular game-theoretic model of learning in games.
In FP, players repeatedly play a game.
At each iteration, each player chooses a best response to their opponents' historical average strategies.
The average strategy profile converges to an NE in certain classes of games, including two-player zero-sum and potential games.
Players may update their beliefs simultaneously or alternately \citep{Berger_2007}.
\citet{leslie2006generalised} introduced \emph{generalised weakened FP (GWFP)}, which generalises FP by allowing approximate best responses and perturbed average strategy updates.

\citet{perkins2014stochastic} extended stochastic FP to the continuous action space framework.
They studied the limiting behaviour of stochastic FP using the associated smooth best response dynamics on the space of finite signed measures.
Using this approach, they showed stochastic FP converges to an equilibrium in two-player zero-sum games.

\citet{Heinrich_2015} introduced \emph{full-width extensive-form FP (XFP)}, which extends FP to extensive-form (multi-step) games.
They also introduced \emph{fictitious self-play (FSP)}, a machine learning framework that implements GWFP in behavioural strategies and in a sample-based fashion.
In FSP, players repeatedly play a game and store their experience in memory.
Instead of playing a best response, they act cautiously and mix between their best responses and average strategies.
At each iteration players replay their experience of play against their opponents to compute an approximate best response.
Similarly, they replay their experience of their own behaviour to learn a model of their average strategy.

\citet{Heinrich16:Deep} introduced \emph{neural fictitious self-play (NFSP)}, which combines FSP with neural network function approximation and deep reinforcement learning.
An NFSP agent consists of two neural networks.
The first network is trained by reinforcement learning from memorized experience of play against fellow agents.
This network learns an approximate best response to the historical behaviour of other agents.
The second network is trained by supervised learning from memorized experience of the agent’s own behaviour.

\citet{Kamra_2019} introduced \emph{DeepFP}, an approximate FP algorithm for two-player games with continuous action spaces.
They demonstrate stable convergence to equilibrium on several classic games and a large forest security domain.
DeepFP represents players' approximate best responses via generative neural networks, which are highly expressive implicit density approximators.
They employ a game-model network that is a differentiable approximation of the players' utilities given their actions, and train these networks end-to-end in a model-based learning regime.
This allows working in the absence of gradients for players.

\citet{Ganzfried_2021} introduced \emph{redundant FP (RFP)}, an algorithm for approximating equilibria in continuous games, and applied it to a continuous Colonel Blotto game.
Unlike ours, this algorithm requires a best-response oracle as a subroutine (\emph{e.g.}, a mixed-integer linear program solver for the continuous Colonel Blotto game).

%% file: additional_figures.tex
\section{Additional figures}

In this section, we include additional figures that did not fit in the body of the paper.

\begin{figure}
    \centering
    \includegraphics[width=.6\linewidth]{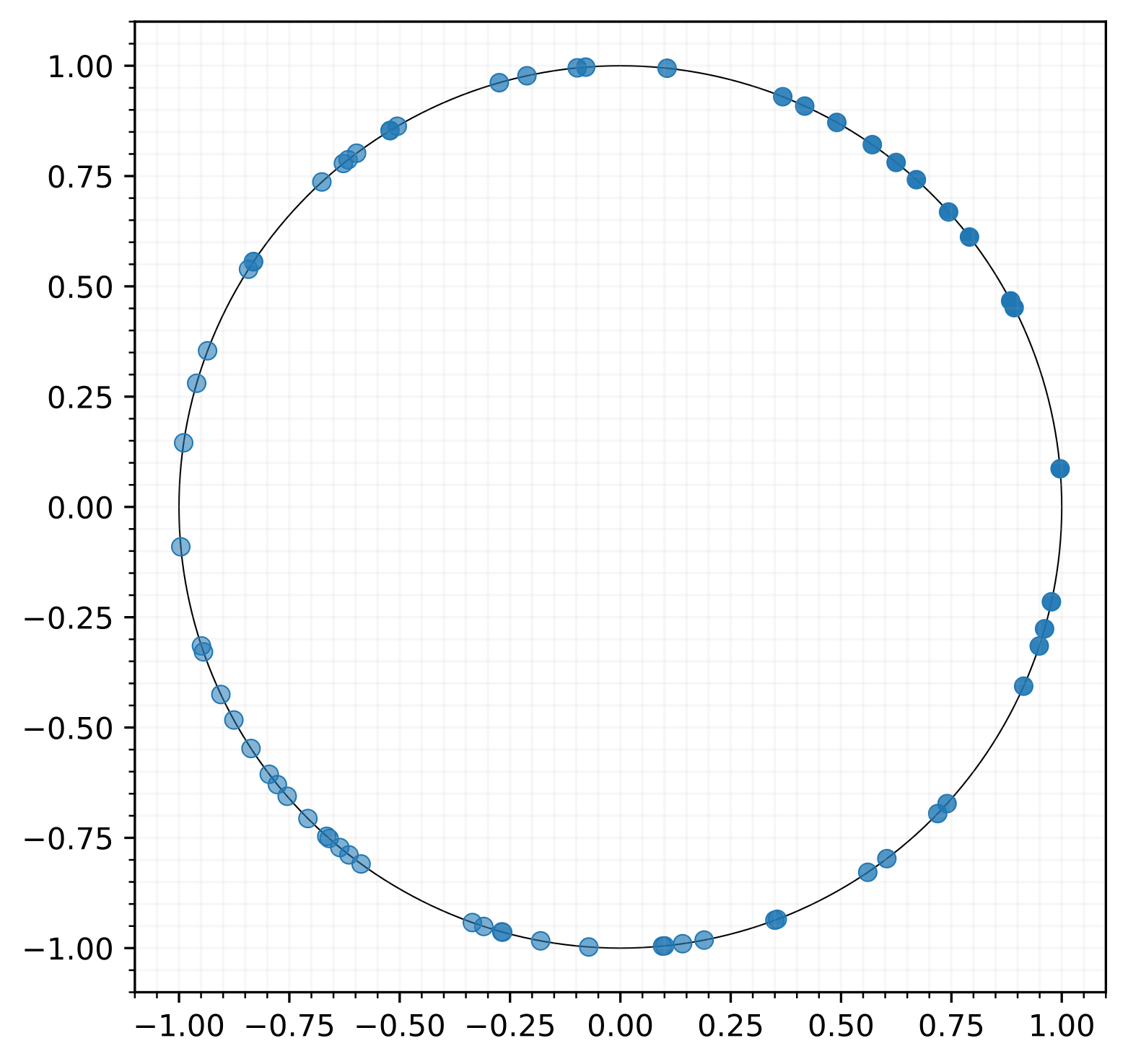}
    \includegraphics[width=.6\linewidth]{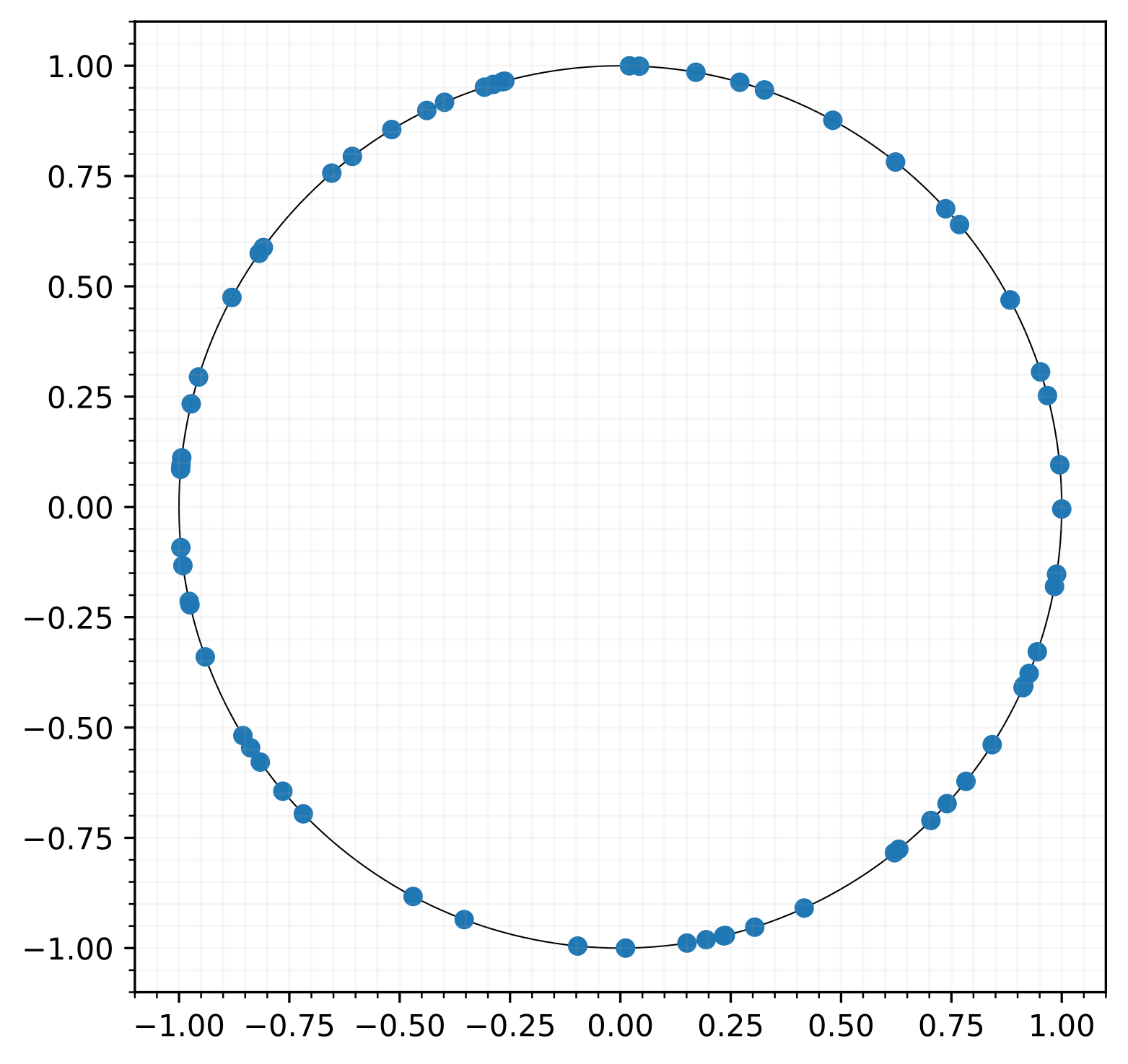}
    \caption{Learned strategies for the circle game.
    Each dot shows an action and its opacity shows its probability.}
    \label{fig:circle_strategies}
\end{figure}

\begin{figure}
    \centering
    \includegraphics[width=.7\linewidth]{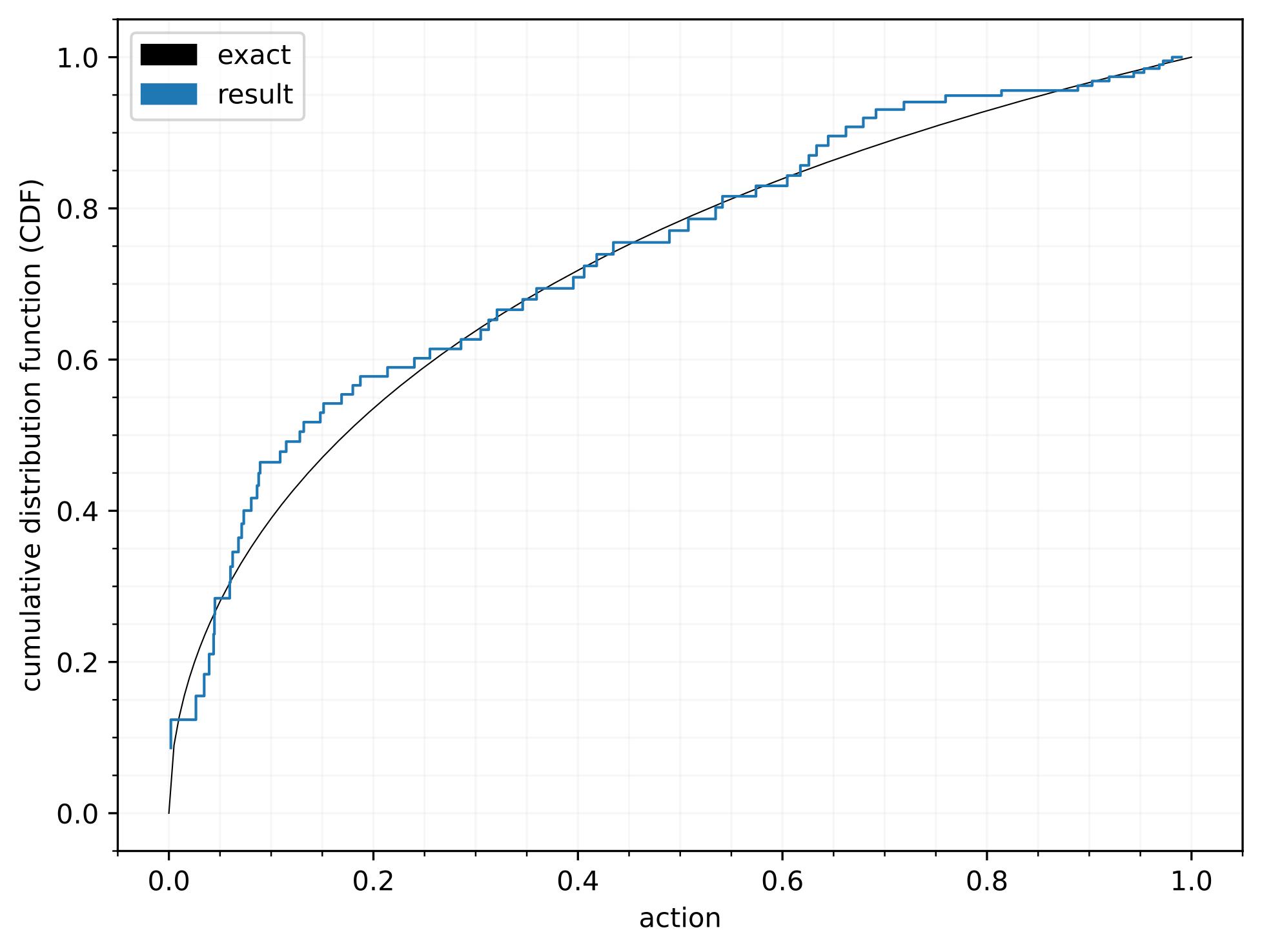}
    \includegraphics[width=.7\linewidth]{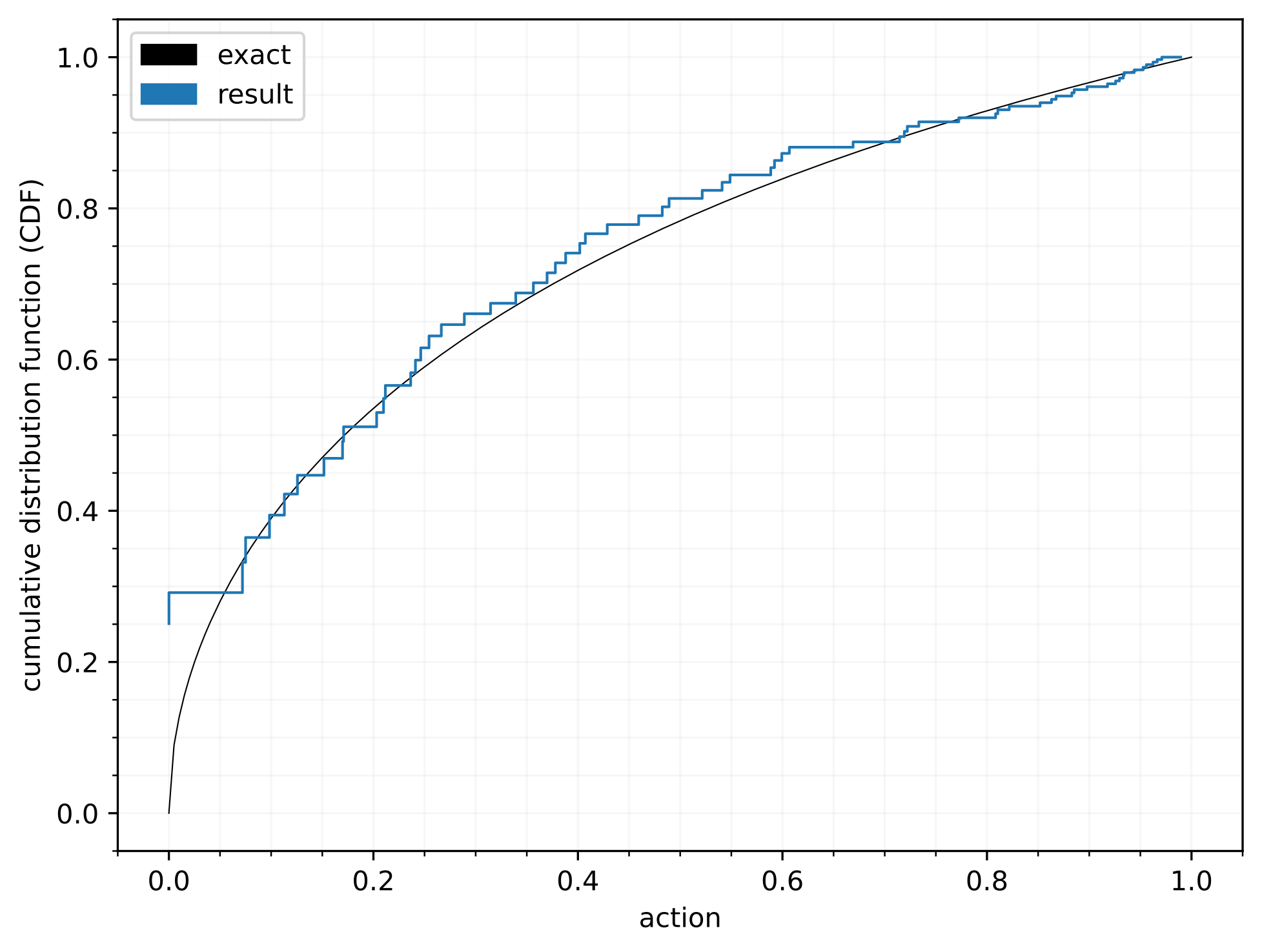}
    \caption{Learned strategies for the Glicksberg--Gross game.
    The exact equilibrium strategies are shown in black.}
    \label{fig:glicksberg_gross_strategies}
\end{figure}

\begin{figure}
    \centering
    \includegraphics[width=.7\linewidth]{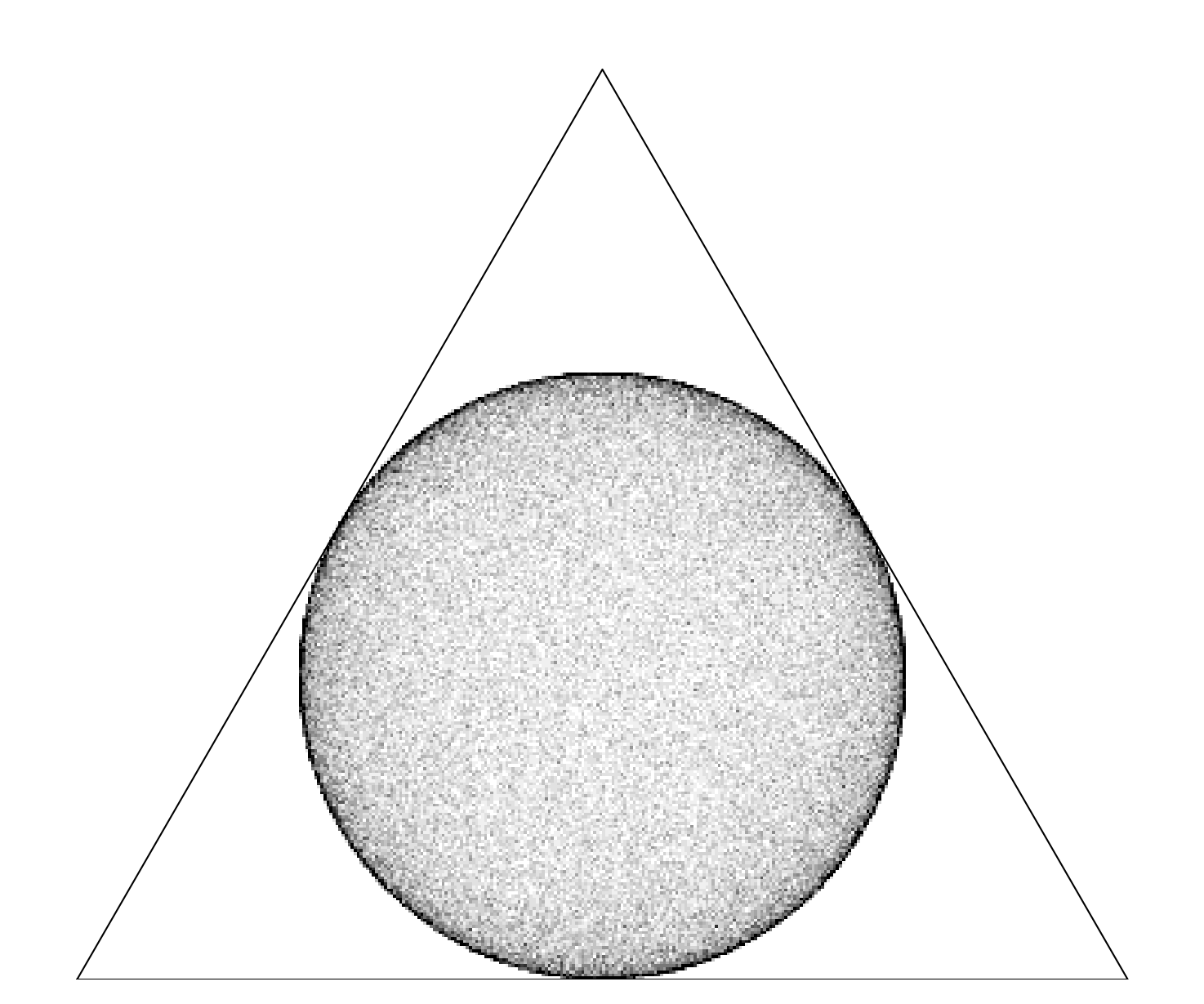}
    \caption{Exact equilibrium for the 2-player, 3-item continuous Colonel Blotto game.
    Darkness shows density.}
    \label{fig:blotto_equilibrium}
\end{figure}

\begin{figure}
    \centering
    \includegraphics[width=.7\linewidth]{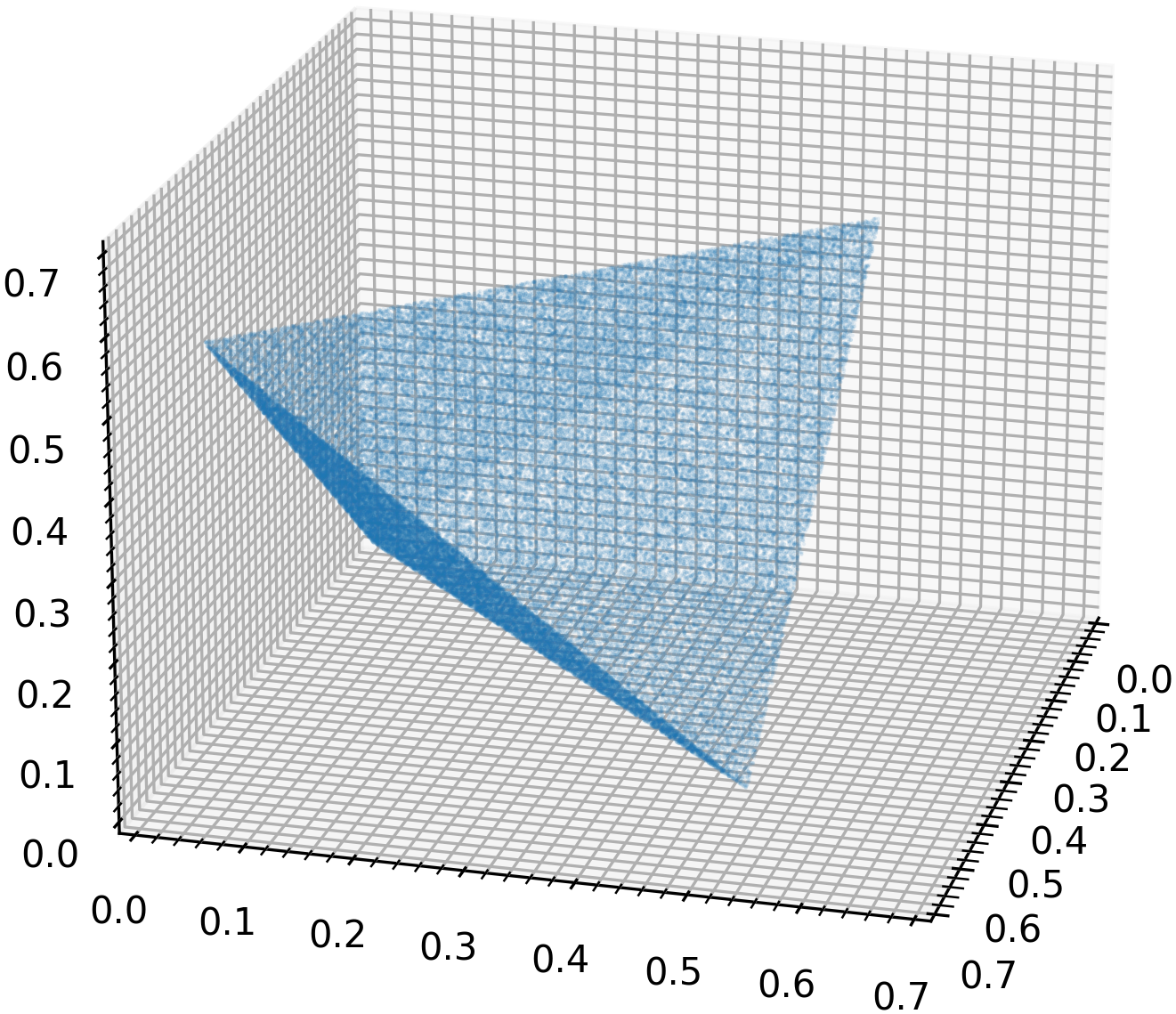}
    \caption{Exact equilibrium for the chopstick auction.}
    \label{fig:chopstick_equilibrium}
\end{figure}

\begin{figure}
    \centering
    \includegraphics[width=.7\linewidth]{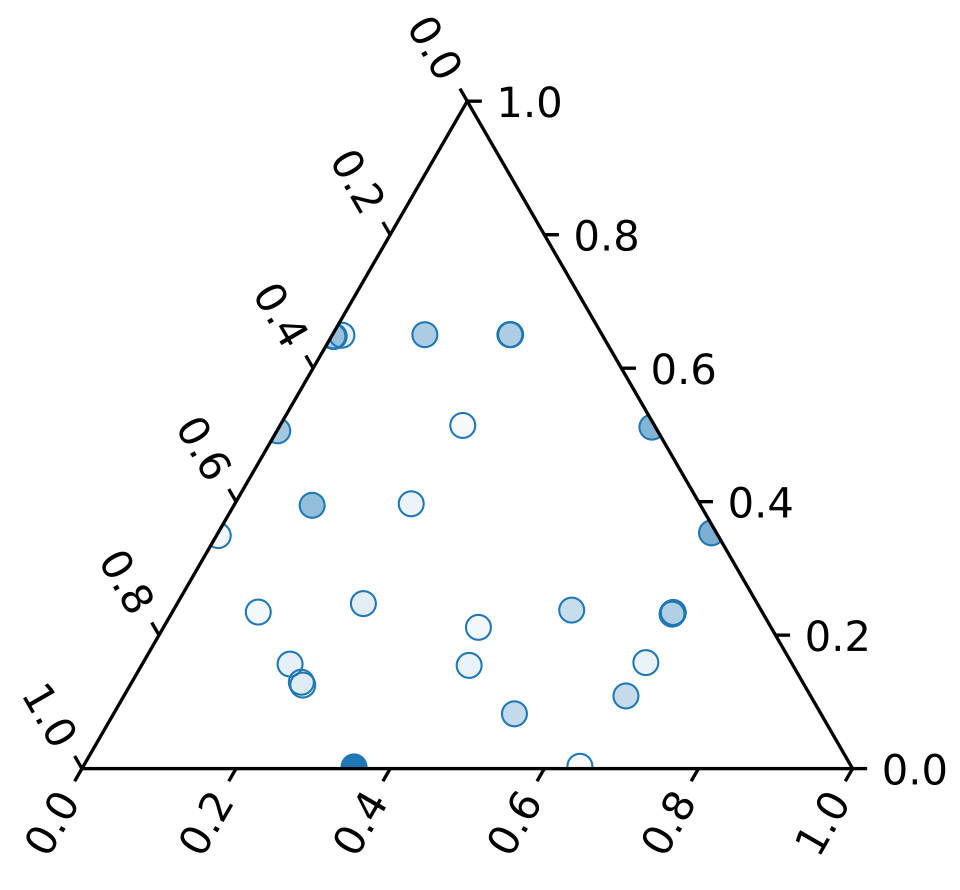}
    \includegraphics[width=.7\linewidth]{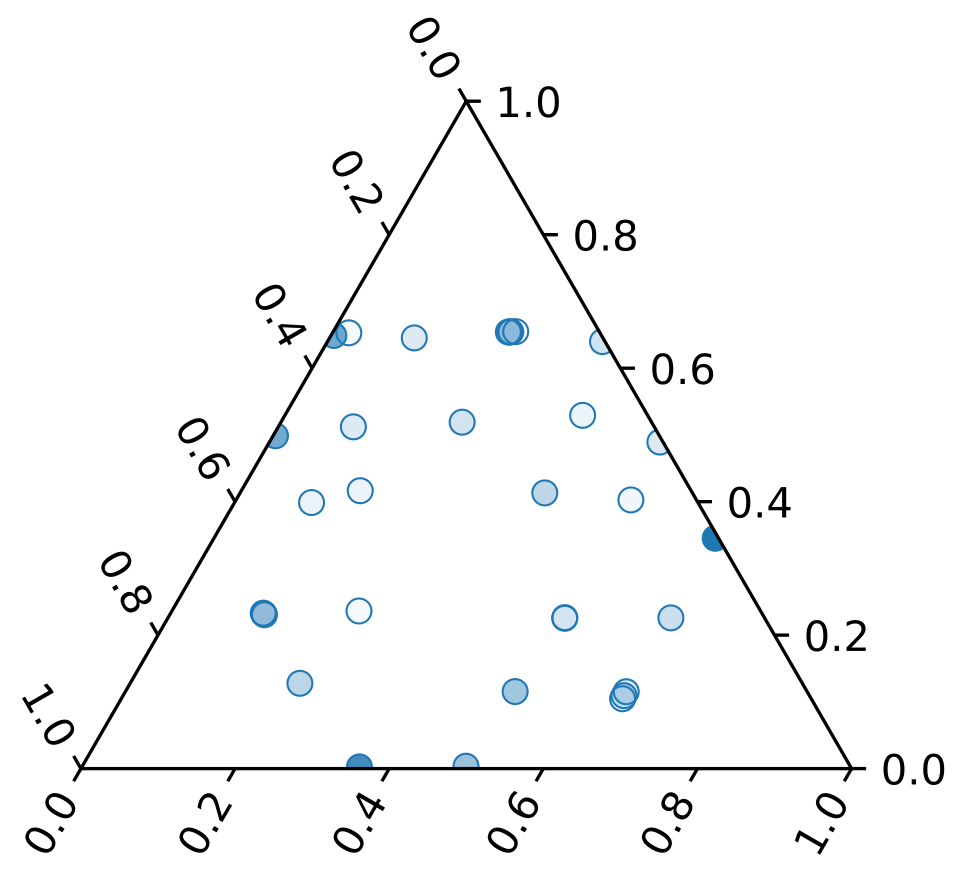}
    \caption{Learned strategies for the 2-player, 3-item continuous Colonel Blotto game.
    Each dot shows an action and its opacity shows its probability.}
    \label{fig:blotto_strategies}
\end{figure}

\begin{figure}
    \centering
    \includegraphics[width=.7\linewidth]{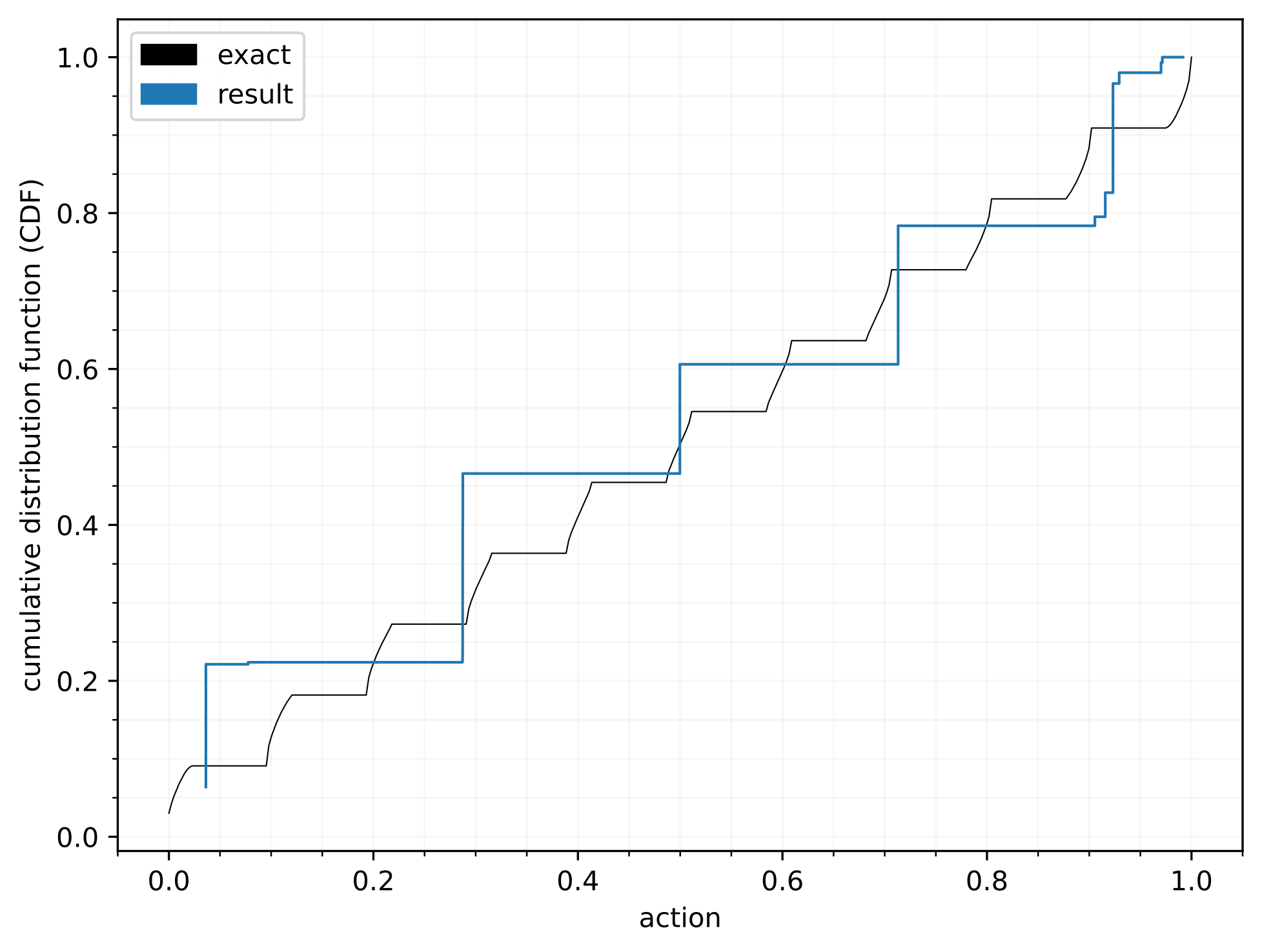}
    \includegraphics[width=.7\linewidth]{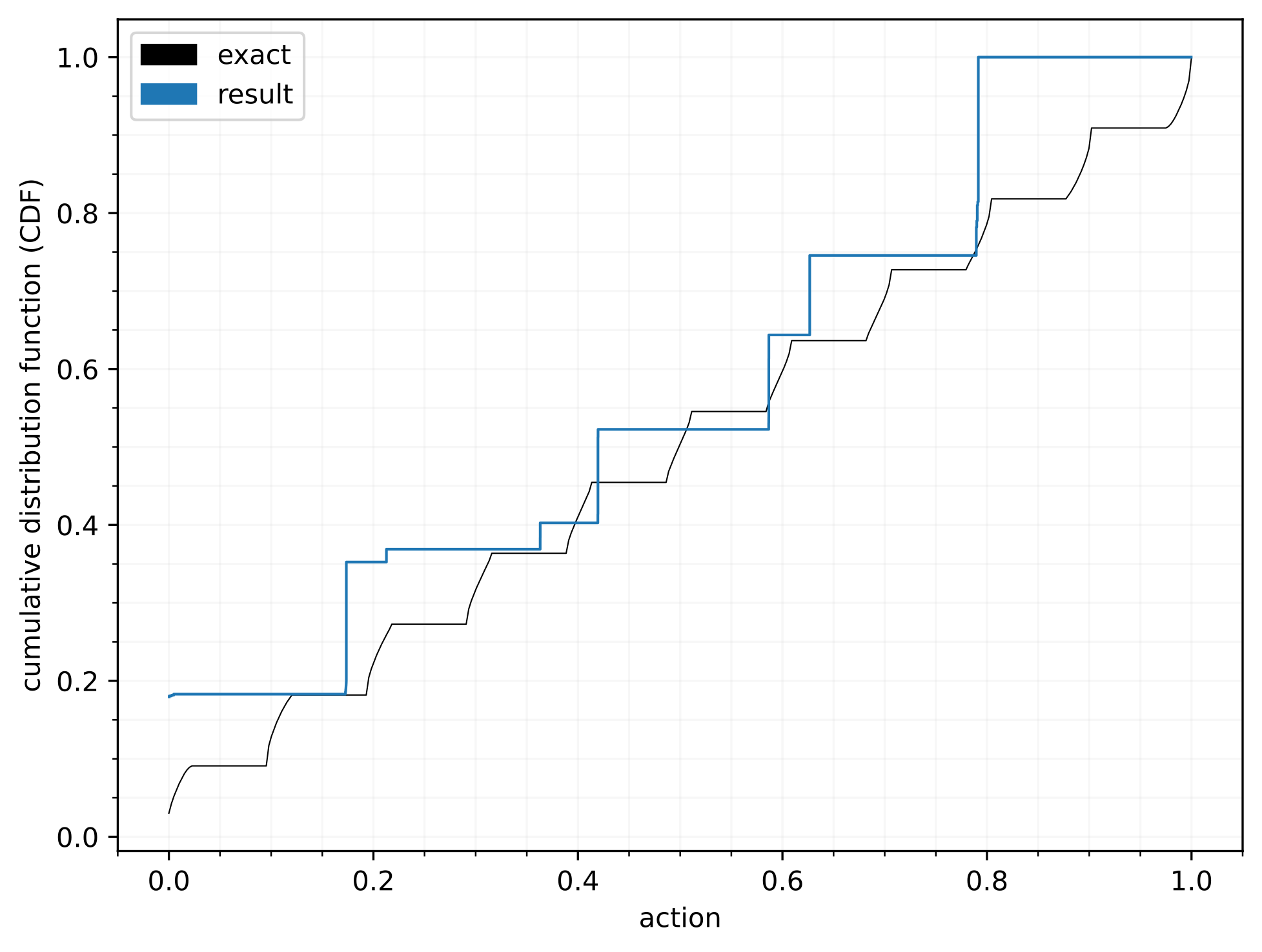}
    \caption{Learned strategies for the discoordination game.
    The exact equilibrium strategies are shown in black.}
    \label{fig:security_strategies}
\end{figure}

\begin{figure}
    \centering
    \includegraphics[width=.7\linewidth]{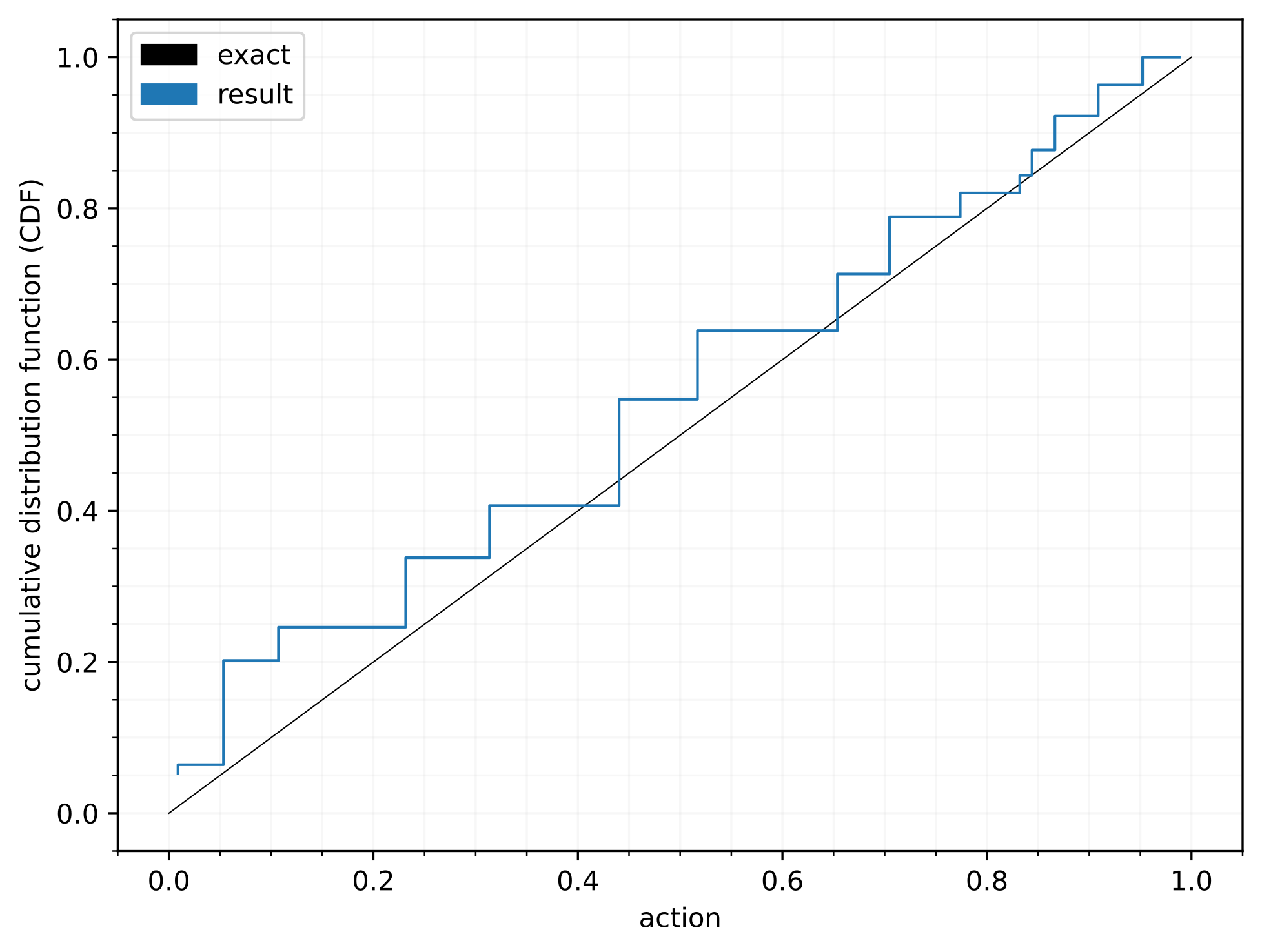}
    \includegraphics[width=.7\linewidth]{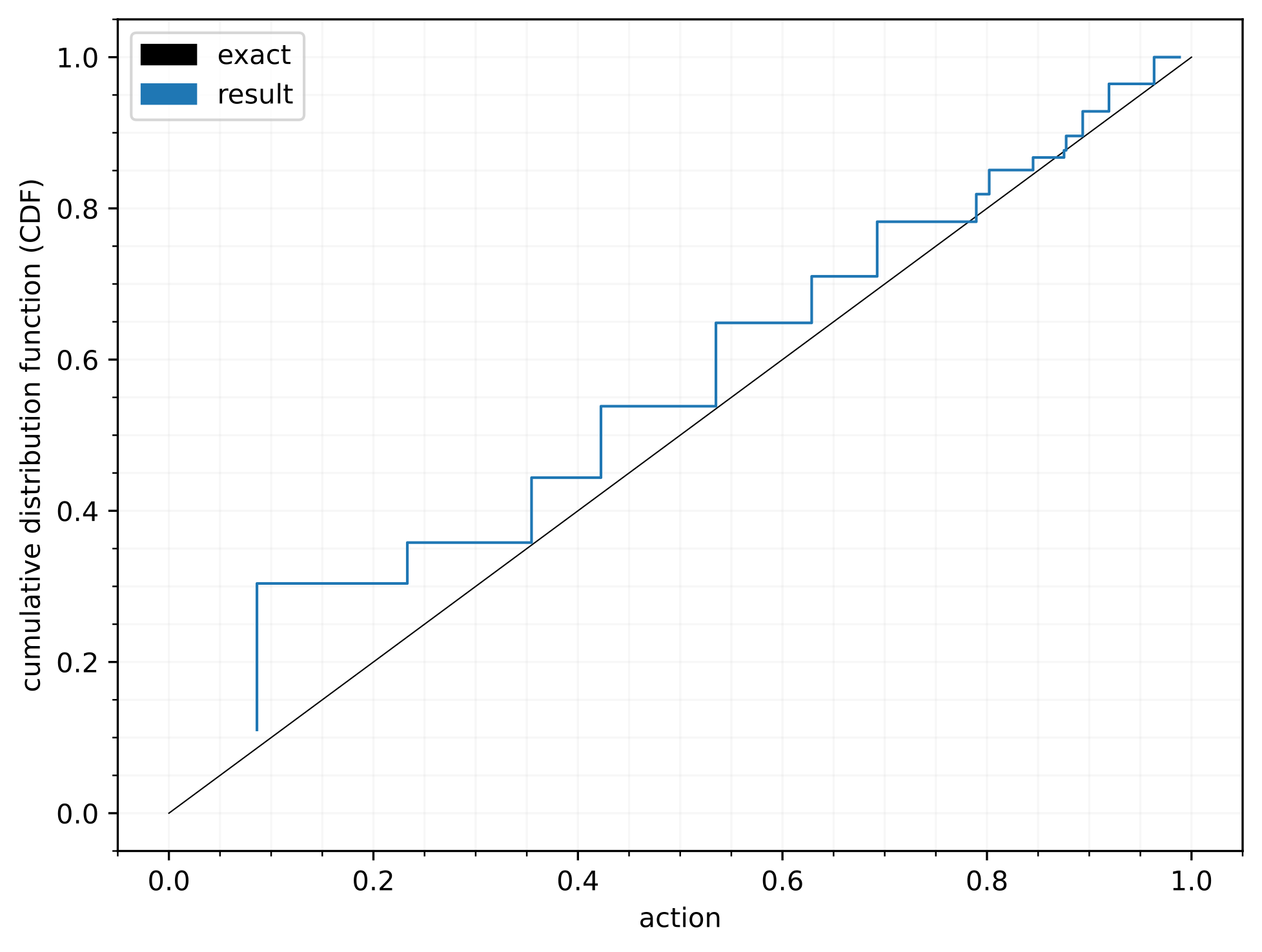}
    \caption{Learned strategies for the 2-player complete-information all-pay auction.
    The exact equilibrium strategies are shown in black.}
    \label{fig:auction_strategies}
\end{figure}

%% file: theoretical/main.tex
\section{Theoretical analysis}
\label{sec:theoretical}

In this section, we present a theoretical analysis of our approach.
A general convergence proof is beyond the scope of this paper, though a potentially interesting question for future research. 
It is not uncommon in this field for methods to be introduced before theoretical guarantees are obtained.
Indeed, the latter is often difficult enough to stand alone as a research contribution. 
There have been many purely theoretical papers in the field of game solving on addressing such open theoretical questions with guarantees or lower bounds.
Furthermore, most of the great breakthroughs in AI game-playing lack theoretical guarantees for the technique that is actually used in practice, especially when they employ neural networks or abstraction techniques.
Theoretical analyses in the literature often make assumptions---such as linearity, (quasi)convexity, \emph{etc.}---that are not always satisfied in practice.

One component of our algorithm is incrementally performing exploitability descent on the metagame at each step.
Assuming the supports are fixed or slow-moving, we can analyze the conditions under which exploitability descent yields an equilibrium for the metagame.

\citet{Lockhart_2019} analyzed exploitability descent in two-player, zero-sum, extensive-form games with finite action spaces.
As stated by \citet{goktas2022exploitability} minimizing exploitability is a logical approach to computing NE, especially in cases where exploitability is convex, such as in the pseudo-games that result from replacing each player's utility function in a monotone pseudo-game by a second-order Taylor approximation \citep{flam1994noncooperative}.

\input{theoretical/convex_exploitability}

\input{theoretical/subgradient_descent}

%% file: theoretical/convex_exploitability.tex
\subsection{Convex exploitability}

In this subsection, we prove that the exploitability function is convex for certain classes of games.

\begin{definition}
Call a game ``regular'' if it has convex strategy sets and its utility function is of the form
\begin{align}
    u_i(x) = f_i(x_i) + \sum_{j \neq i} g(x_i, x_j)
\end{align}
where \(g_{ij}\) is convex in its second argument.
\end{definition}

\begin{theorem}
\label{thm:regular}
A constant-sum regular game has convex exploitability.
\end{theorem}
\begin{proof}
A sum of convex functions is convex.
A supremum of convex functions is convex.
Therefore,
\begin{align}
    \Phi(x)
    &= \sum_{i \in \mathcal{I}} \mleft( \sup_{y_i \in \mathcal{S}_i} u_i(y_i, x_{-i}) - u_i(x) \mright) \\
    &= \sum_{i \in \mathcal{I}} \sup_{y_i \in \mathcal{S}_i} u_i(y_i, x_{-i}) - \sum_{i \in \mathcal{I}} u_i(x) \\
    &= \sum_{i \in \mathcal{I}} \sup_{y_i \in \mathcal{S}_i} u_i(y_i, x_{-i}) + \text{const} \\
    &= \sum_{i \in \mathcal{I}} \sup_{y_i \in \mathcal{S}_i} \mleft( f_i(y_i) + \sum_{j \neq i} g(y_i, x_j) \mright) + \text{const} \\
    &= \sum_{i \in \mathcal{I}} \sup_{y_i \in \mathcal{S}_i} \mleft( \text{const} + \sum_{j \neq i} \text{convex} \mright) + \text{const} \\
    &= \sum_{i \in \mathcal{I}} \sup_{y_i \in \mathcal{S}_i} \text{convex} + \text{const} \\
    &= \sum_{i \in \mathcal{I}} \text{convex} + \text{const} \\
    &= \text{convex}
\end{align}
\end{proof}

\begin{definition}
A polymatrix game is a game with a utility function of the form
\begin{align}
    u_i(x) = \sum_{j \neq i} x_i^\top A_{ij} x_j
\end{align}
\end{definition}

These are graphical games in which each node corresponds to a player and each edge corresponds to a two-player bimatrix game between its endpoints.
Each player chooses a single strategy for all of its bimatrix games and receives the sum of the resulting payoffs.
In a \emph{constant-sum} polymatrix game, the sum of utilities across all players is constant.
As noted by \citet{cai2011minmax}, ``Intuitively, these games can be used to model a broad class of competitive environments where there is a constant amount of wealth (resources) to be split among the players of the game, with no in-flow or out-flow of wealth that may change the total sum of players’ wealth in an outcome of the game.''
They give an example of a ``Wild West'' game in which a set of gold miners need to transport gold by splitting it into wagons that traverse different paths, each of which may be controlled by thieves who could seize it.

\citet{cai2011minmax} prove a generalization of von Neumann's minmax theorem to constant-sum polymatrix games.
Their theorem implies convexity of equilibria, polynomial-time tractability, and convergence of no-regret learning algorithms to Nash equilibria.
\citet{cai2016zero} show that, in such games,  Nash equilibria can be found efficiently with linear programming.
They also show that the set of coarse correlated equilibria (CCE) collapses to the set of Nash equilibria.

We prove the following result.

\begin{theorem}
A constant-sum polymatrix game has convex exploitability.
\end{theorem}
\begin{proof}
A polymatrix game is a regular game where \(f_i(x_i) = 0\) and \(g_{ij}(x_i, x_j) = x_i^\top A_{ij} x_j\).
The latter is linear, and therefore convex, in its second argument.
Therefore, if the game is constant-sum, by Theorem \ref{thm:regular}, the exploitability is convex.
\end{proof}

\begin{corollary}
A pairwise constant-sum polymatrix game has convex exploitability.
\end{corollary}

\begin{corollary}
A two-player constant-sum matrix game has convex exploitability.
\end{corollary}

\begin{theorem}
A two-player constant-sum concave-convex game has convex exploitability.
\end{theorem}
\begin{proof}
In a two-player constant-sum game, the exploitability reduces to the so-called \textit{duality gap} \citep{Grnarova_2021}.
\begin{align}
    \Phi(x) = \sup_{x_1' \in \mathcal{S}_1} u_1(x_1', x_2) - \inf_{x_2' \in \mathcal{S}_2} u_1(x_1, x_2') + C
\end{align}
Here, \(u_1(x_1', x_2)\) is convex in \(x\).
Thus \(\sup_{x_1' \in \mathcal{S}_1} u_1(x_1', x_2)\) is convex in \(x\).
Also, \(u_1(x_1, x_2')\) is concave in \(x\).
Thus \(\inf_{x_2' \in \mathcal{S}_2} u_1(x_1, x_2')\) is concave in \(x\).
Thus \(\Phi(x) = \text{convex} - \text{concave} + \text{constant}\) in \(x\), which is convex in \(x\).
\end{proof}

%% file: theoretical/subgradient_descent.tex
\subsection{Subgradient descent}

We minimize exploitability by performing subgradient descent.
This raises the question of when this process attains a global minimum.
\citet{kiwiel2004convergence} analyze the convergence of \emph{approximate subgradient methods} for convex optimization, and prove the following theorems.
Let \(\mathcal{S} \subseteq \mathbb{R}^n\) be a nonempty closed convex set,
\(f : \mathcal{S} \to \mathbb{R}\) be a closed proper convex function,
and \(\mathcal{S}_* = \operatorname{argmin} f\).
Let \(x_{t+1} = P_\mathcal{S}(x_t - \nu_t g_t)\) where \(P_\mathcal{S}\) is the projector onto \(\mathcal{S}\) (\(P_\mathcal{S}(x) \in \argmin_{y \in \mathcal{S}} \|x - y\|\)), \(\nu_t \geq 0\) is a stepsize, \(\varepsilon_t \geq 0\) is an error tolerance, and \(g_t \in \partial_{\varepsilon_t} f(x_t)\) is an \(\varepsilon_t\)-\emph{approximate subgradient} of \(f\) at \(x_t\), that is, \(f(x) \geq f(x_t) + \langle g_t, x - x_t \rangle - \varepsilon_t\) for all \(x\).

\begin{theorem}
\citep[Theorem 3.4]{kiwiel2004convergence}
Suppose \(\mathcal{S}_* \neq \varnothing\), \(\sum_{t \in \mathbb{N}} \nu_t = \infty\), and \(\sum_{t \in \mathbb{N}} \nu_t (\tfrac{1}{2} \|g_t\|^2 \nu_t + \varepsilon_t) < \infty\).
Then \(\{x_t\}_{t \in \mathbb{N}}\) converges to some \(x_\infty \in \mathcal{S}_*\).
\end{theorem}

\begin{theorem}
\citep[Theorem 3.6]{kiwiel2004convergence}
Suppose \(\mathcal{S}_* \neq \varnothing\), \(\sum_{t \in \mathbb{N}} \nu_t = \infty\), \(\sum_{t \in \mathbb{N}} \nu_t^2 < \infty\), \(\sum_{t \in \mathbb{N}} \nu_t \varepsilon_t < \infty\), and the subgradients do not grow too fast: \(\exists c < \infty . \forall t \in \mathbb{N} . \|g_t\|^2 \leq c (1 + \|x_t\|^2)\) (\emph{e.g.}, they are bounded).
Then \(\{x_t\}_{t \in \mathbb{N}}\) converges to some \(x_\infty \in \mathcal{S}_*\).
\end{theorem}

Similar convergence results are known for subgradient methods on \emph{quasi}convex functions \citep{hu2015inexact}.
Therefore, when the assumptions of the above theorems hold, \(\{x_t\}_{t \in \mathbb{N}}\) converges to a global minimizer of the exploitability function, which is an NE (if an NE exists at all).

%% file: code.tex
\section{Code}

We use Python 3.12.2 with the following libraries:
\begin{itemize}
    \item \texttt{jax} 0.4.28 \citep{jax2018github}: \url{https://github.com/google/jax}
    \item \texttt{flax} 0.8.3 \citep{flax2020github}: \url{https://github.com/google/flax}
    \item \texttt{optax} 0.2.2 \citep{deepmind2020jax}: \url{https://github.com/google-deepmind/optax}
    \item \texttt{matplotlib} 3.8.4 \citep{Hunter:2007}: \url{https://github.com/matplotlib/matplotlib}
\end{itemize}
An implementation of our method is shown below.

\begin{lstlisting}[
language=python,
breaklines=true,
basicstyle=\ttfamily\scriptsize,
postbreak=\mbox{\textcolor{red}{\(\hookrightarrow\)}\space},
breakindent=0em,
breakatwhitespace=true,
]
from functools import partial

import jax
import optax
from jax import nn, numpy as jnp

from lib import nfg_utils, utils

def mix(x, axis=None, keepdims=False):
  ranks = 1 + x.argsort(axis)
  weights = ranks / ranks.max(axis, keepdims=True)
  return (x * weights).sum(axis, keepdims=keepdims)

class StochasticGradientAscent:
  def __init__(self, utility_fn, optimizer, grad=jax.grad):
    self.utility_fn = utility_fn
    self.optimizer = optimizer
    self.grad = grad

  def init(self, params, key):
    opt_state = self.optimizer.init(params)
    return params, opt_state

  def grads(self, params, key):
    return self.grad(self.utility_fn)(params, key)

  def step(self, state, key):
    params, opt_state = state
    grads = self.grads(params, key)
    grads = jax.tree.map(jnp.negative, grads)
    updates, opt_state = self.optimizer.update(grads, opt_state, params)
    params = optax.apply_updates(params, updates)
    return (params, opt_state), state

  def get_params(self, state):
    params, opt_state = state
    return params

class SISAMS(StochasticGradientAscent):
  def __init__(self, utility_fn, optimizer):
    self.utility_fn = utility_fn
    self.optimizer = optimizer

  def init(self, supports, key):
    weights = [
      jnp.zeros(utils.tree_batch_dim(support)) for support in supports
    ]
    params = supports, weights
    opt_state = self.optimizer.init(params)
    return params, opt_state

  def grads(self, params, key):
    utility_fn = partial(self.utility_fn, key=key)

    supports, weights = params

    def exploitability(weights):
      tensor = nfg_utils.make_tensor(utility_fn, supports)
      strategies = list(map(nn.softmax, weights))
      return nfg_utils.get_exploitability(tensor, strategies)

    weights_grad = jax.grad(exploitability)(weights)
    weights_grad = jax.tree.map(jnp.negative, weights_grad)

    def value(support, player):
      new_supports = utils.replace(supports, player, support)
      tensor = nfg_utils.make_tensor(utility_fn, new_supports)
      strategies = list(map(nn.softmax, weights))
      vector = nfg_utils.get_player_utilities(tensor, strategies, player)
      return mix(vector)

    supports_grad = [
      jax.grad(value)(support, player)
      for player, support in enumerate(supports)
    ]

    return supports_grad, weights_grad

  def get_params(self, state):
    (supports, weights), opt_state = state
    return supports, weights
\end{lstlisting}

%% file: main.bbl
\begin{thebibliography}{98}
\providecommand{\natexlab}[1]{#1}
\providecommand{\url}[1]{\texttt{#1}}
\expandafter\ifx\csname urlstyle\endcsname\relax
  \providecommand{\doi}[1]{doi: #1}\else
  \providecommand{\doi}{doi: \begingroup \urlstyle{rm}\Url}\fi

\bibitem[Adam et~al.(2021)Adam, Hor{\v{c}}{\'\i}k, Kasl, and Kroupa]{Adam_2021}
Luk{\'a}{\v{s}} Adam, Rostislav Hor{\v{c}}{\'\i}k, Tom{\'a}{\v{s}} Kasl, and Tom{\'a}{\v{s}} Kroupa.
\newblock Double oracle algorithm for computing equilibria in continuous games.
\newblock \emph{AAAI Conference on Artificial Intelligence (AAAI)}, 35, 2021.

\bibitem[Adamo and Matros(2009)]{Adamo_2009}
Tim Adamo and Alexander Matros.
\newblock A {B}lotto game with incomplete information.
\newblock \emph{Economics Letters}, 105, 2009.

\bibitem[Aumann(1974)]{Aumann74:Subjectivity}
Robert Aumann.
\newblock Subjectivity and correlation in randomized strategies.
\newblock \emph{Journal of Mathematical Economics}, 1:\penalty0 67--96, 1974.

\bibitem[Baye et~al.(1996)Baye, Kovenock, and de~Vries]{Baye_1996}
Michael~R. Baye, Dan Kovenock, and Casper~G. de~Vries.
\newblock The all-pay auction with complete information.
\newblock \emph{Economic Theory}, 8, 1996.

\bibitem[Berahas et~al.(2022)Berahas, Cao, Choromanski, and Scheinberg]{Berahas_2022}
Albert~S. Berahas, Liyuan Cao, Krzysztof Choromanski, and Katya Scheinberg.
\newblock A theoretical and empirical comparison of gradient approximations in derivative-free optimization.
\newblock \emph{Foundations of Computational Mathematics}, 22, 2022.

\bibitem[Berger(2007)]{Berger_2007}
Ulrich Berger.
\newblock {B}rown's original fictitious play.
\newblock \emph{Journal of Economic Theory}, 135, 2007.

\bibitem[Berner et~al.(2019)Berner, Brockman, Chan, Cheung, Debiak, Dennison, Farhi, Fischer, Hashme, Hesse, et~al.]{Berner19:Dota}
Christopher Berner, Greg Brockman, Brooke Chan, Vicki Cheung, Przemys{\l}aw Debiak, Christy Dennison, David Farhi, Quirin Fischer, Shariq Hashme, Chris Hesse, et~al.
\newblock Dota 2 with large scale deep reinforcement learning.
\newblock \emph{arXiv preprint arXiv:1912.06680}, 2019.

\bibitem[Bichler et~al.(2021)Bichler, Fichtl, Heidekr{\"u}ger, Kohring, and Sutterer]{Bichler_2021}
Martin Bichler, Maximilian Fichtl, Stefan Heidekr{\"u}ger, Nils Kohring, and Paul Sutterer.
\newblock Learning equilibria in symmetric auction games using artificial neural networks.
\newblock \emph{Nature Machine Intelligence}, 3, 2021.

\bibitem[Bichler et~al.(2022)Bichler, Fichtl, and Oberlechner]{Fichtl_2022}
Martin Bichler, Max Fichtl, and Matthias Oberlechner.
\newblock Computing {B}ayes {N}ash equilibrium strategies in auction games via simultaneous online dual averaging.
\newblock \emph{arXiv:2208.02036}, 2022.

\bibitem[Bichler et~al.(2023)Bichler, Kohring, and Heidekr{\"u}ger]{Bichler_2022}
Martin Bichler, Nils Kohring, and Stefan Heidekr{\"u}ger.
\newblock Learning equilibria in asymmetric auction games.
\newblock \emph{INFORMS Journal on Computing (IJOC)}, 2023.

\bibitem[Boix-Adserà et~al.(2021)Boix-Adserà, Edelman, and Jayanti]{Adsera_2021}
Enric Boix-Adserà, Benjamin~L. Edelman, and Siddhartha Jayanti.
\newblock The multiplayer {C}olonel {B}lotto game.
\newblock \emph{Games and Economic Behavior}, 129, 2021.

\bibitem[Borel(1938)]{Borel38:Traite}
{\'E}mile Borel.
\newblock \emph{Trait\'{e} du calcul des probabilit\'{e}s et ses applications}, volume~IV of \emph{Applications aux jeux des hazard}.
\newblock Gauthier-Villars, Paris, 1938.

\bibitem[Borel(1953)]{Borel_1953}
{\'E}mile Borel.
\newblock The theory of play and integral equations with skew symmetric kernels.
\newblock \emph{Econometrica}, 21, 1953.

\bibitem[Bradbury et~al.(2018)Bradbury, Frostig, Hawkins, Johnson, Leary, Maclaurin, Necula, Paszke, VanderPlas, Wanderman-Milne, et~al.]{jax2018github}
James Bradbury, Roy Frostig, Peter Hawkins, Matthew~James Johnson, Chris Leary, Dougal Maclaurin, George Necula, Adam Paszke, Jake VanderPlas, Skye Wanderman-Milne, et~al.
\newblock {JAX}: Composable transformations of {P}ython+{N}um{P}y programs, 2018.

\bibitem[Brown(1951)]{Brown51:Iterative}
George~W. Brown.
\newblock Iterative solutions of games by fictitious play.
\newblock In Tjalling~C. Koopmans, editor, \emph{Activity Analysis of Production and Allocation}, pages 374--376. John Wiley \& Sons, 1951.

\bibitem[Brown and Sandholm(2014)]{Brown14:Regret}
Noam Brown and Tuomas Sandholm.
\newblock Regret transfer and parameter optimization.
\newblock In \emph{AAAI Conference on Artificial Intelligence (AAAI)}, 2014.

\bibitem[Brown and Sandholm(2015)]{Brown15:Simultaneous}
Noam Brown and Tuomas Sandholm.
\newblock Simultaneous abstraction and equilibrium finding in games.
\newblock In \emph{Proceedings of the International Joint Conference on Artificial Intelligence (IJCAI)}, 2015.

\bibitem[Cai and Daskalakis(2011)]{cai2011minmax}
Yang Cai and Constantinos Daskalakis.
\newblock On minmax theorems for multiplayer games.
\newblock In \emph{Proceedings of the twenty-second annual ACM-SIAM symposium on Discrete algorithms}, pages 217--234. SIAM Journal on Computing, 2011.

\bibitem[Cai et~al.(2016)Cai, Candogan, Daskalakis, and Papadimitriou]{cai2016zero}
Yang Cai, Ozan Candogan, Constantinos Daskalakis, and Christos Papadimitriou.
\newblock Zero-sum polymatrix games: A generalization of minmax.
\newblock \emph{Mathematics of Operations Research}, 41\penalty0 (2):\penalty0 648--655, 2016.

\bibitem[Chen and Ankenman(2006)]{Chen06:Mathematics}
Bill Chen and Jerrod Ankenman.
\newblock \emph{The Mathematics of Poker}.
\newblock ConJelCo, 2006.

\bibitem[Chin et~al.(1976)Chin, Parthasarathy, and Raghavan]{Chin_1976}
H.~Chin, Thiruvenkatachari Parthasarathy, and Thirukkannamangai Raghavan.
\newblock Optimal strategy sets for continuous two person games.
\newblock \emph{Sankhyā: The Indian Journal of Statistics, Series A (1961-2002)}, 38, 1976.

\bibitem[Cho et~al.(2014)Cho, Merri{\"e}nboer, Bahdanau, and Bengio]{Cho_2014}
Kyunghyun Cho, Bart~Van Merri{\"e}nboer, Dzmitry Bahdanau, and Yoshua Bengio.
\newblock On the properties of neural machine translation: encoder-decoder approaches.
\newblock In \emph{Eighth Workshop on Syntax, Semantics and Structure in Statistical Translation (SSST-8)}, 2014.

\bibitem[Dasgupta and Maskin(1986)]{Dasgupta86:Existence}
P.~Dasgupta and Eric Maskin.
\newblock The existence of equilibrium in discontinuous economic games 1: Theory.
\newblock \emph{Review of Economic Studies}, 53:\penalty0 1--26, 1986.

\bibitem[Debreu(1952)]{Debreu_1952}
Gerard Debreu.
\newblock A social equilibrium existence theorem.
\newblock \emph{Proceedings of the National Academy of Sciences}, 38, 1952.

\bibitem[DeepMind et~al.(2020)DeepMind, Babuschkin, Baumli, Bell, Bhupatiraju, Bruce, Buchlovsky, Budden, Cai, Clark, Danihelka, Dedieu, Fantacci, Godwin, Jones, Hemsley, Hennigan, Hessel, Hou, Kapturowski, Keck, Kemaev, King, Kunesch, Martens, Merzic, Mikulik, Norman, Papamakarios, Quan, Ring, Ruiz, Sanchez, Sartran, Schneider, Sezener, Spencer, Srinivasan, Stanojevi\'{c}, Stokowiec, Wang, Zhou, and Viola]{deepmind2020jax}
DeepMind, Igor Babuschkin, Kate Baumli, Alison Bell, Surya Bhupatiraju, Jake Bruce, Peter Buchlovsky, David Budden, Trevor Cai, Aidan Clark, Ivo Danihelka, Antoine Dedieu, Claudio Fantacci, Jonathan Godwin, Chris Jones, Ross Hemsley, Tom Hennigan, Matteo Hessel, Shaobo Hou, Steven Kapturowski, Thomas Keck, Iurii Kemaev, Michael King, Markus Kunesch, Lena Martens, Hamza Merzic, Vladimir Mikulik, Tamara Norman, George Papamakarios, John Quan, Roman Ring, Francisco Ruiz, Alvaro Sanchez, Laurent Sartran, Rosalia Schneider, Eren Sezener, Stephen Spencer, Srivatsan Srinivasan, Milo\v{s} Stanojevi\'{c}, Wojciech Stokowiec, Luyu Wang, Guangyao Zhou, and Fabio Viola.
\newblock The {D}eep{M}ind {JAX} {E}cosystem, 2020.
\newblock URL \url{http://github.com/google-deepmind}.

\bibitem[Duchi et~al.(2015)Duchi, Jordan, Wainwright, and Wibisono]{Duchi_2015}
John~C. Duchi, Michael~I. Jordan, Martin~J. Wainwright, and Andre Wibisono.
\newblock Optimal rates for zero-order convex optimization: The power of two function evaluations.
\newblock \emph{IEEE Transactions on Information Theory}, 61, 2015.

\bibitem[Fan(1952)]{Fan_1952}
Ky~Fan.
\newblock Fixed point and minimax theorems in locally convex topological linear spaces.
\newblock \emph{Proceedings of the National Academy of Sciences}, 38, 1952.

\bibitem[Flam and Ruszczynski(1994)]{flam1994noncooperative}
S.~D. Flam and Andrzej Ruszczynski.
\newblock Noncooperative convex games: computing equilibrium by partial regularization.
\newblock Technical report, International Institute for Applied Systems Analysis, 1994.

\bibitem[Ganzfried(2021)]{Ganzfried_2021}
Sam Ganzfried.
\newblock Algorithm for computing approximate {N}ash equilibrium in continuous games with application to continuous {B}lotto.
\newblock \emph{Games}, 12, 2021.

\bibitem[Ganzfried and Sandholm(2010)]{Ganzfried10:Computing}
Sam Ganzfried and Tuomas Sandholm.
\newblock Computing equilibria by incorporating qualitative models.
\newblock In \emph{International Conference on Autonomous Agents and Multi-Agent Systems (AAMAS)}, 2010.

\bibitem[Ghosh and Kundu(2019)]{Ghosh_2019}
Papiya Ghosh and Rajendra~P. Kundu.
\newblock Best-shot network games with continuous action space.
\newblock \emph{Research in Economics}, 73, 2019.

\bibitem[Glicksberg(1952)]{Glicksberg52:Further}
I.~L. Glicksberg.
\newblock A further generalization of the {K}akutani fixed point theorem, with application to {N}ash equilibrium points.
\newblock \emph{Proceedings of the American Mathematical Society}, 3\penalty0 (1):\penalty0 170--174, 1952.

\bibitem[Glicksberg and Gross(1953)]{Glicksberg_1953}
Irving~Leonard Glicksberg and Oliver~Alfred Gross.
\newblock Notes on games over the square.
\newblock \emph{Contributions to the theory of games}, 2, 1953.

\bibitem[Goktas and Greenwald(2022)]{goktas2022exploitability}
Denizalp Goktas and Amy Greenwald.
\newblock Exploitability minimization in games and beyond.
\newblock \emph{Conference on Neural Information Processing Systems (NeurIPS)}, 35, 2022.

\bibitem[Goodfellow et~al.(2014)Goodfellow, Pouget-Abadie, Mirza, Xu, Warde-Farley, Ozair, Courville, and Bengio]{Goodfellow_2014}
Ian Goodfellow, Jean Pouget-Abadie, Mehdi Mirza, Bing Xu, David Warde-Farley, Sherjil Ozair, Aaron Courville, and Yoshua Bengio.
\newblock Generative adversarial nets.
\newblock In \emph{Conference on Neural Information Processing Systems (NeurIPS)}, 2014.

\bibitem[Goodfellow et~al.(2020)Goodfellow, Pouget-Abadie, Mirza, Xu, Warde-Farley, Ozair, Courville, and Bengio]{goodfellow2020generative}
Ian Goodfellow, Jean Pouget-Abadie, Mehdi Mirza, Bing Xu, David Warde-Farley, Sherjil Ozair, Aaron Courville, and Yoshua Bengio.
\newblock Generative adversarial networks.
\newblock \emph{Communications of the ACM}, 63, 2020.

\bibitem[Grnarova et~al.(2021)Grnarova, Kilcher, Levy, Lucchi, and Hofmann]{Grnarova_2021}
Paulina Grnarova, Yannic Kilcher, Kfir~Y. Levy, Aurelien Lucchi, and Thomas Hofmann.
\newblock Generative minimization networks: Training {GAN}s without competition.
\newblock \emph{arXiv:2103.12685}, 2021.

\bibitem[Gross(1957)]{gross1957rational}
Oliver~Alfred Gross.
\newblock A rational game on the square.
\newblock \emph{Contributions to the theory of games}, 3, 1957.

\bibitem[Gross and Wagner(1950)]{Gross_1950}
Oliver~Alfred Gross and R.~A. Wagner.
\newblock \emph{A continuous {C}olonel {B}lotto game}.
\newblock RAND Corporation, 1950.

\bibitem[Heek et~al.(2023)Heek, Levskaya, Oliver, Ritter, Rondepierre, Steiner, and van {Z}ee]{flax2020github}
Jonathan Heek, Anselm Levskaya, Avital Oliver, Marvin Ritter, Bertrand Rondepierre, Andreas Steiner, and Marc van {Z}ee.
\newblock {F}lax: A neural network library and ecosystem for {JAX}, 2023.
\newblock URL \url{http://github.com/google/flax}.

\bibitem[Heinrich and Silver(2016)]{Heinrich16:Deep}
Johannes Heinrich and David Silver.
\newblock Deep reinforcement learning from self-play in imperfect-information games.
\newblock \emph{arXiv preprint arXiv:1603.01121}, 2016.

\bibitem[Heinrich et~al.(2015)Heinrich, Lanctot, and Silver]{Heinrich_2015}
Johannes Heinrich, Marc Lanctot, and David Silver.
\newblock Fictitious self-play in extensive-form games.
\newblock In \emph{International Conference on Machine Learning (ICML)}, volume~37, 2015.

\bibitem[Hochreiter and Schmidhuber(1997)]{Hochreiter_1997}
Sepp Hochreiter and J\"{u}rgen Schmidhuber.
\newblock Long short-term memory.
\newblock \emph{Neural Computation}, 9, 1997.

\bibitem[Hu et~al.(2015)Hu, Yang, and Sim]{hu2015inexact}
Yaohua Hu, Xiaoqi Yang, and Chee-Khian Sim.
\newblock Inexact subgradient methods for quasi-convex optimization problems.
\newblock \emph{European Journal of Operational Research}, 240, 2015.

\bibitem[Hunter(2007)]{Hunter:2007}
J.~D. Hunter.
\newblock Matplotlib: A 2d graphics environment.
\newblock \emph{Computing in Science \& Engineering}, 9:\penalty0 90--95, 2007.

\bibitem[Kamra et~al.(2017)Kamra, Fang, Kar, Liu, and Tambe]{Kamra_2017}
Nitin Kamra, Fei Fang, Debarun Kar, Yan Liu, and Milind Tambe.
\newblock Handling continuous space security games with neural networks.
\newblock In \emph{Proceedings of the International Joint Conference on Artificial Intelligence (IJCAI)}, 2017.

\bibitem[Kamra et~al.(2018)Kamra, Gupta, Fang, Liu, and Tambe]{Kamra_2018}
Nitin Kamra, Umang Gupta, Fei Fang, Yan Liu, and Milind Tambe.
\newblock Policy learning for continuous space security games using neural networks.
\newblock \emph{AAAI Conference on Artificial Intelligence (AAAI)}, 32, 2018.

\bibitem[Kamra et~al.(2019)Kamra, Gupta, Wang, Fang, Liu, and Tambe]{Kamra_2019}
Nitin Kamra, Umang Gupta, Kai Wang, Fei Fang, Yan Liu, and Milind Tambe.
\newblock {DeepFP} for finding {N}ash equilibrium in continuous action spaces.
\newblock In \emph{Decision and Game Theory for Security}, 2019.

\bibitem[Kiwiel(2004)]{kiwiel2004convergence}
Krzysztof~C. Kiwiel.
\newblock Convergence of approximate and incremental subgradient methods for convex optimization.
\newblock \emph{SIAM Journal on Optimization (SIOPT)}, 14, 2004.

\bibitem[Kovenock and Roberson(2011)]{Kovenock_2011}
Dan Kovenock and Brian Roberson.
\newblock A {B}lotto game with multi-dimensional incomplete information.
\newblock \emph{Economics Letters}, 113, 2011.

\bibitem[Kovenock and Roberson(2021)]{Kovenock_2021}
Dan Kovenock and Brian Roberson.
\newblock Generalizations of the {G}eneral {L}otto and {C}olonel {B}lotto games.
\newblock \emph{Economic Theory}, 71, 2021.

\bibitem[Kroer and Sandholm(2014)]{Kroer14:Extensive}
Christian Kroer and Tuomas Sandholm.
\newblock Extensive-form game abstraction with bounds.
\newblock In \emph{Proceedings of the ACM Conference on Economics and Computation (EC)}, 2014.

\bibitem[Kroer and Sandholm(2015)]{Kroer15:Discretization}
Christian Kroer and Tuomas Sandholm.
\newblock Discretization of continuous action spaces in extensive-form games.
\newblock In \emph{International Conference on Autonomous Agents and Multi-Agent Systems (AAMAS)}, 2015.

\bibitem[Kroer and Sandholm(2016)]{Kroer16:Imperfect}
Christian Kroer and Tuomas Sandholm.
\newblock Imperfect-recall abstractions with bounds in games.
\newblock In \emph{Proceedings of the ACM Conference on Economics and Computation (EC)}, 2016.

\bibitem[Kroupa and Votroubek(2023)]{Kroupa_2021}
Tom{\'a}{\v{s}} Kroupa and Tom{\'a}{\v{s}} Votroubek.
\newblock Multiple oracle algorithm to solve continuous games.
\newblock In \emph{International Conference on Decision and Game Theory for Security}, 2023.

\bibitem[Kuhn and Tucker(1953)]{Kuhn_Tucker_1953}
Harold~William Kuhn and Albert~William Tucker.
\newblock \emph{Contributions to the theory of games}, volume~2.
\newblock Princeton University Press, 1953.

\bibitem[Lanctot et~al.(2017)Lanctot, Zambaldi, Gruslys, Lazaridou, Tuyls, P{\'e}rolat, Silver, and Graepel]{Lanctot17:Unified}
Marc Lanctot, Vinicius Zambaldi, Audrunas Gruslys, Angeliki Lazaridou, Karl Tuyls, Julien P{\'e}rolat, David Silver, and Thore Graepel.
\newblock A unified game-theoretic approach to multiagent reinforcement learning.
\newblock In \emph{Conference on Neural Information Processing Systems (NeurIPS)}, pages 4190--4203, 2017.

\bibitem[Leslie and Collins(2006)]{leslie2006generalised}
David~S. Leslie and Edmund~J. Collins.
\newblock Generalised weakened fictitious play.
\newblock \emph{Games and Economic Behavior}, 56, 2006.

\bibitem[Li and Wellman(2021)]{Li_2021}
Zun Li and Michael~P. Wellman.
\newblock Evolution strategies for approximate solution of {B}ayesian games.
\newblock \emph{AAAI Conference on Artificial Intelligence (AAAI)}, 35, 2021.

\bibitem[Lockhart et~al.(2019)Lockhart, Lanctot, P{\'e}rolat, Lespiau, Morrill, Timbers, and Tuyls]{Lockhart_2019}
Edward Lockhart, Marc Lanctot, Julien P{\'e}rolat, Jean-Baptiste Lespiau, Dustin Morrill, Finbarr Timbers, and Karl Tuyls.
\newblock Computing approximate equilibria in sequential adversarial games by exploitability descent.
\newblock In \emph{Proceedings of the International Joint Conference on Artificial Intelligence (IJCAI)}, 2019.

\bibitem[Marchesi et~al.(2020)Marchesi, Trov{\`o}, and Gatti]{Marchesi20:Learning}
Alberto Marchesi, Francesco Trov{\`o}, and Nicola Gatti.
\newblock Learning probably approximately correct maximin strategies in simulation-based games with infinite strategy spaces.
\newblock In \emph{Autonomous Agents and Multi-Agent Systems}, pages 834--842, 2020.

\bibitem[Marris et~al.(2021)Marris, Muller, Lanctot, Tuyls, and Graepel]{marris2021multi}
Luke Marris, Paul Muller, Marc Lanctot, Karl Tuyls, and Thore Graepel.
\newblock Multi-agent training beyond zero-sum with correlated equilibrium meta-solvers.
\newblock In \emph{International Conference on Machine Learning (ICML)}, 2021.

\bibitem[Martin and Sandholm(2023)]{ijcai2023p317}
Carlos Martin and Tuomas Sandholm.
\newblock Finding mixed-strategy equilibria of continuous-action games without gradients using randomized policy networks.
\newblock In \emph{Proceedings of the International Joint Conference on Artificial Intelligence (IJCAI)}, 2023.

\bibitem[McAleer et~al.(2020)McAleer, Lanier, Fox, and Baldi]{mcaleer2020pipeline}
Stephen McAleer, John~B. Lanier, Roy Fox, and Pierre Baldi.
\newblock Pipeline {PSRO}: A scalable approach for finding approximate {N}ash equilibria in large games.
\newblock \emph{Conference on Neural Information Processing Systems (NeurIPS)}, 33, 2020.

\bibitem[McAleer et~al.(2021)McAleer, Lanier, Wang, Baldi, and Fox]{McAleer_2021}
Stephen McAleer, John~B. Lanier, Kevin~A. Wang, Pierre Baldi, and Roy Fox.
\newblock {XDO}: A double oracle algorithm for extensive-form games.
\newblock In \emph{Conference on Neural Information Processing Systems (NeurIPS)}, volume~34, 2021.

\bibitem[McAleer et~al.(2022{\natexlab{a}})McAleer, Lanier, Wang, Baldi, Fox, and Sandholm]{McAleer_2022b}
Stephen McAleer, John~B. Lanier, Kevin Wang, Pierre Baldi, Roy Fox, and Tuomas Sandholm.
\newblock Self-play {PSRO}: Toward optimal populations in two-player zero-sum games.
\newblock \emph{arXiv:2207.06541}, 2022{\natexlab{a}}.

\bibitem[McAleer et~al.(2022{\natexlab{b}})McAleer, Wang, Lanier, Lanctot, Baldi, Sandholm, and Fox]{McAleer_2022}
Stephen McAleer, Kevin Wang, John Lanier, Marc Lanctot, Pierre Baldi, Tuomas Sandholm, and Roy Fox.
\newblock Anytime {PSRO} for two-player zero-sum games.
\newblock \emph{arXiv:2201.07700}, 2022{\natexlab{b}}.

\bibitem[McMahan et~al.(2003)McMahan, Gordon, and Blum]{McMahan_2003}
H.~Brendan McMahan, Geoffrey~J. Gordon, and Avrim Blum.
\newblock Planning in the presence of cost functions controlled by an adversary.
\newblock In \emph{International Conference on Machine Learning (ICML)}, 2003.

\bibitem[Metz et~al.(2021)Metz, Freeman, Schoenholz, and Kachman]{metz2021gradients}
Luke Metz, C.~Daniel Freeman, Samuel~S. Schoenholz, and Tal Kachman.
\newblock Gradients are not all you need.
\newblock \emph{arXiv:2111.05803}, 2021.

\bibitem[Milgrom and Segal(2014)]{Milgrom14:Deferred}
Paul Milgrom and Ilya Segal.
\newblock Deferred-acceptance auctions and radio spectrum reallocation.
\newblock In \emph{Proceedings of the ACM Conference on Economics and Computation (EC)}, 2014.

\bibitem[Milgrom and Segal(2020)]{Milgrom_2020}
Paul Milgrom and Ilya Segal.
\newblock Clock auctions and radio spectrum reallocation.
\newblock \emph{Journal of Political Economy}, 2020.

\bibitem[Milgrom and Weber(1985)]{Milgrom85:Distributional}
Paul Milgrom and Robert Weber.
\newblock Distributional strategies for games with incomplete infromation.
\newblock \emph{Mathematics of Operations Research}, 10:\penalty0 619--632, 1985.

\bibitem[Muller et~al.(2020)Muller, Omidshafiei, Rowland, Tuyls, Perolat, Liu, Hennes, Marris, Lanctot, Hughes, et~al.]{Muller2020A}
Paul Muller, Shayegan Omidshafiei, Mark Rowland, Karl Tuyls, Julien Perolat, Siqi Liu, Daniel Hennes, Luke Marris, Marc Lanctot, Edward Hughes, et~al.
\newblock A generalized training approach for multiagent learning.
\newblock In \emph{International Conference on Learning Representations (ICLR)}, 2020.

\bibitem[Nash(1950)]{Nash50:Equilibrium}
John Nash.
\newblock Equilibrium points in n-person games.
\newblock \emph{Proceedings of the National Academy of Sciences}, 36:\penalty0 48--49, 1950.

\bibitem[Nesterov and Spokoiny(2017)]{Nesterov_2017}
Yurii Nesterov and Vladimir Spokoiny.
\newblock Random gradient-free minimization of convex functions.
\newblock \emph{Foundations of Computational Mathematics}, 17, 2017.

\bibitem[Omidshafiei et~al.(2019)Omidshafiei, Papadimitriou, Piliouras, Tuyls, Rowland, Lespiau, Czarnecki, Lanctot, Perolat, and Munos]{omidshafiei2019alpha}
Shayegan Omidshafiei, Christos Papadimitriou, Georgios Piliouras, Karl Tuyls, Mark Rowland, Jean-Baptiste Lespiau, Wojciech~M. Czarnecki, Marc Lanctot, Julien Perolat, and Remi Munos.
\newblock $\alpha$-rank: Multi-agent evaluation by evolution.
\newblock \emph{Scientific Reports}, 9, 2019.

\bibitem[Parrilo(2006)]{parrilo2006polynomial}
Pablo~A. Parrilo.
\newblock Polynomial games and sum of squares optimization.
\newblock In \emph{IEEE Conference on Decision and Control (CDC)}, 2006.

\bibitem[Parthasarathy(1970)]{Parthasarathy_1970}
Thiruvenkatachari Parthasarathy.
\newblock On games over the unit square.
\newblock \emph{SIAM Journal on Applied Mathematics}, 19, 1970.

\bibitem[Parthasarathy and Raghavan(1975)]{parthasarathy1975equilibria}
Thiruvenkatachari Parthasarathy and Thirukkannamangai Raghavan.
\newblock Equilibria of continuous two-person games.
\newblock \emph{Pacific Journal of Mathematics}, 57, 1975.

\bibitem[Perkins and Leslie(2014)]{perkins2014stochastic}
Steven Perkins and David~S. Leslie.
\newblock Stochastic fictitious play with continuous action sets.
\newblock \emph{Journal of Economic Theory}, 152, 2014.

\bibitem[Perkins et~al.(2015)Perkins, Mertikopoulos, and Leslie]{perkins2015mixed}
Steven Perkins, Panayotis Mertikopoulos, and David~S. Leslie.
\newblock Mixed-strategy learning with continuous action sets.
\newblock \emph{IEEE Transactions on Automatic Control (TACON)}, 62, 2015.

\bibitem[Rosen(1965)]{Rosen_1965}
J.~Ben Rosen.
\newblock Existence and uniqueness of equilibrium points for concave n-person games.
\newblock \emph{Econometrica}, 33, 1965.

\bibitem[Rumelhart et~al.(1986)Rumelhart, Hinton, and Williams]{Rumelhart_1986}
David~E. Rumelhart, Geoffrey~E. Hinton, and Ronald~J. Williams.
\newblock Learning representations by back-propagating errors.
\newblock \emph{Nature}, 323, 1986.

\bibitem[Salimans et~al.(2017)Salimans, Ho, Chen, Sidor, and Sutskever]{Salimans_2017}
Tim Salimans, Jonathan Ho, Xi~Chen, Szymon Sidor, and Ilya Sutskever.
\newblock Evolution strategies as a scalable alternative to reinforcement learning.
\newblock \emph{arXiv:1703.03864}, 2017.

\bibitem[Sandholm(2013)]{Sandholm13:Very}
Tuomas Sandholm.
\newblock Very-large-scale generalized combinatorial multi-attribute auctions: Lessons from conducting \$60 billion of sourcing.
\newblock In Zvika Neeman, Alvin Roth, and Nir Vulkan, editors, \emph{Handbook of Market Design}. Oxford University Press, 2013.

\bibitem[Shamir(2017)]{Shamir_2017}
Ohad Shamir.
\newblock An optimal algorithm for bandit and zero-order convex optimization with two-point feedback.
\newblock \emph{Journal of Machine Learning Research}, 18, 2017.

\bibitem[Stein et~al.(2008)Stein, Ozdaglar, and Parrilo]{stein2008separable}
Noah~D. Stein, Asuman Ozdaglar, and Pablo~A. Parrilo.
\newblock Separable and low-rank continuous games.
\newblock \emph{International Journal of Game Theory (IJGT)}, 37, 2008.

\bibitem[Steinberger et~al.(2020)Steinberger, Lerer, and Brown]{steinberger2020dream}
Eric Steinberger, Adam Lerer, and Noam Brown.
\newblock Dream: Deep regret minimization with advantage baselines and model-free learning.
\newblock \emph{arXiv:2006.10410}, 2020.

\bibitem[Szentes and Rosenthal(2003{\natexlab{a}})]{Szentes03:Chopsticks}
Balazs Szentes and Robert~W. Rosenthal.
\newblock Beyond chopsticks: Symmetric equilibria in majority auction games.
\newblock \emph{Games and Economic Behavior}, 45\penalty0 (2):\penalty0 278--295, 2003{\natexlab{a}}.

\bibitem[Szentes and Rosenthal(2003{\natexlab{b}})]{Szentes_2003}
Balázs Szentes and Robert~W. Rosenthal.
\newblock Three-object two-bidder simultaneous auctions: chopsticks and tetrahedra.
\newblock \emph{Games and Economic Behavior}, 2003{\natexlab{b}}.

\bibitem[Sz{\'e}p et~al.(1985)Sz{\'e}p, Forg{\'o}, Sz{\'e}p, and Forg{\'o}]{szep1985games}
J.~Sz{\'e}p, F.~Forg{\'o}, J.~Sz{\'e}p, and F.~Forg{\'o}.
\newblock Games played over the unit square.
\newblock \emph{Introduction to the theory of games}, 1985.

\bibitem[Timbers et~al.(2022)Timbers, Bard, Lockhart, Lanctot, Schmid, Burch, Schrittwieser, Hubert, and Bowling]{Timbers_2022}
Finbarr Timbers, Nolan Bard, Edward Lockhart, Marc Lanctot, Martin Schmid, Neil Burch, Julian Schrittwieser, Thomas Hubert, and Michael Bowling.
\newblock Approximate exploitability: Learning a best response.
\newblock In \emph{Proceedings of the International Joint Conference on Artificial Intelligence (IJCAI)}, 2022.

\bibitem[Vinyals et~al.(2019)Vinyals, Babuschkin, Czarnecki, Mathieu, Dudzik, Chung, Choi, Powell, Ewalds, Georgiev, et~al.]{Vinyals19:Grandmaster}
Oriol Vinyals, Igor Babuschkin, Wojciech~M Czarnecki, Micha{\"e}l Mathieu, Andrew Dudzik, Junyoung Chung, David~H Choi, Richard Powell, Timo Ewalds, Petko Georgiev, et~al.
\newblock Grandmaster level in {StarCraft II} using multi-agent reinforcement learning.
\newblock \emph{Nature}, 575\penalty0 (7782):\penalty0 350--354, 2019.

\bibitem[Walton and Lisy(2021)]{Walton_2021}
Michael Walton and Viliam Lisy.
\newblock Multi-agent reinforcement learning in {O}pen{S}piel: A reproduction report.
\newblock \emph{arXiv:2103.00187}, 2021.

\bibitem[Washburn(2013)]{Washburn_2013}
Alan Washburn.
\newblock {OR} forum - {B}lotto politics.
\newblock \emph{Operations Research}, 61, 2013.

\bibitem[Werbos(1988)]{Werbos_1988}
Paul~J. Werbos.
\newblock Generalization of backpropagation with application to a recurrent gas market model.
\newblock \emph{Neural networks}, 1, 1988.

\bibitem[Wierstra et~al.(2008)Wierstra, Schaul, Glasmachers, Sun, Peters, and Schmidhuber]{Wierstra_2008}
Daan Wierstra, Tom Schaul, Tobias Glasmachers, Yi~Sun, Jan Peters, and J{\"u}rgen Schmidhuber.
\newblock Natural evolution strategies.
\newblock In \emph{IEEE Congress on Evolutionary Computation}, 2008.

\bibitem[Wierstra et~al.(2014)Wierstra, Schaul, Glasmachers, Sun, Peters, and Schmidhuber]{Wierstra_2014}
Daan Wierstra, Tom Schaul, Tobias Glasmachers, Yi~Sun, Jan Peters, and J{\"u}rgen Schmidhuber.
\newblock Natural evolution strategies.
\newblock \emph{Journal of Machine Learning Research}, 15, 2014.

\end{thebibliography}
